\documentclass[preprint,11pt,sort&compress]{elsarticle}
\usepackage[colorlinks,linkcolor=blue,anchorcolor=black,citecolor=black]{hyperref}
\usepackage{geometry}
\usepackage{amsmath,amsthm,amssymb,amsfonts,bm,bbm,mathrsfs,mathtools}
\usepackage{xfrac}
\newtheorem{theorem}{Theorem}

\newtheorem{remark}{Remark}
\allowdisplaybreaks[4]
\usepackage[ruled,vlined]{algorithm2e}
\usepackage{algorithmic}
\usepackage{tabularx,caption,float,subfig,booktabs}
\usepackage{graphicx}
\usepackage{threeparttable}
\usepackage{multirow}
\usepackage{makecell}
\newcolumntype{Y}{>{\centering\arraybackslash}X}

\begin{document}
\begin{frontmatter}
\title{Random batch list method for metallic system with embedded atom potential}
\author[NWU]{Jieqiong~Zhang}  %
\author[LSEC,UCAS]{Jizu~Huang\corref{cor1}}
\cortext[cor1]{Corresponding author.}
\ead{huangjz@lsec.cc.ac.cn}
\author[NWPU,SH]{Zihao~Yang\corref{cor1}}
\ead{yangzihao@nwpu.edu.cn}

\address[NWU]{School of mathematics, Northwest University, Xi'an, 710127, China.}
\address[LSEC]{LSEC, Academy of Mathematics and Systems Science, Chinese Academyof Sciences, Beijing, 100190, China}
\address[UCAS]{School of Mathematical Sciences, University  of Chinese Academy of Sciences, Beijing, 100049, China.}
\address[NWPU]{School of mathematics and statistics, Northwestern Polytechnical University, Xi'an, 710072, China.}
\address[SH]{Shenzhen Research Institute of Northwestern Polytechnical University, Shenzhen, 518063, China}

\journal{Elsevier}

\begin{abstract}
The embedded atom method (EAM) is one of the most widely used many-body, short-range potentials in molecular dynamics simulations, particularly for metallic systems. To enhance the efficiency of calculating these short-range interactions, we extend the random batch list (RBL) concept to the EAM potential, resulting in the RBL-EAM algorithm. The newly presented method introduces two “core-shell” lists for approximately computing the host electron densities and the force terms, respectively. 
Direct interactions are computed in the core regions, while in the shell zones a random batch list is used to reduce the number of interaction pairs, leading to significant reductions in both computational complexity and storage requirements. We provide a theoretical, unbiased estimate of the host electron densities and the force terms. Since metallic systems are Newton-pair systems, we extend the RBL-EAM algorithm to exploit this property, thereby halving the computational cost. Numerical examples, including the lattice constant, the radial distribution function, and the elastic constants, demonstrate that the RBL-EAM method significantly accelerates simulations several times without compromising accuracy.

\end{abstract}
    
    
\end{frontmatter}

\section{Introduction}
%
In recent decades, the molecular dynamics (MD) method \cite{zhou2022MD,frenkel2023MD,hollingsworth2018MD} has achieved remarkable success in various fields such as chemical physics, materials science, and biophysics \cite{hospital2015MD,wang2021MD,krishna2021MD}, emerging as a widely used computational simulation technique to analyze the physical movements of atoms and molecules. 
By numerically integrating Newton's equations of motion, MD predicts the time evolution of atomic trajectories within a molecular system and enables the measurement of the equilibrium and dynamical properties of physical systems based on the resulting atomic configurations \cite{tuckerman2023statistical}.
The force terms in Newton's equations are determined by the interactions between particles, which are often calculated using interatomic potentials or molecular mechanical force fields. 
For nonbonded interactions, such as Lennard-Jones (LJ) \cite{heinz2008LJ}, the embedded atom method (EAM) \cite{daw1993embedded,daw1984embedded}, and Coulomb potentials \cite{pasichnyk2004coulomb}, the computational cost associated with force calculations typically accounts for 90\% or more of the total computational effort \cite{hollingsworth2018MD}. This substantial computational burden significantly limits both the efficiency and the scalability of the MD method.

Coulomb interactions, a typical example of long-range interactions, are commonly computed using Ewald-type summation methods \cite{ewald1921berechnung}. These methods handle the long-range smooth component in a uniform mesh using fast Fourier transforms, while the short-range contribution is computed in real space \cite{belhadj1991Ewald, deserno1998Ewald}. 
Ewald-type methods reduce computational complexity from ${\cal O}(N^2)$ to ${\cal O}(N\log N)$ or even ${\cal O}(N)$ and have been successfully applied in various contexts \cite{deserno1998Ewald,plimpton1997Ewald,batcho2001Ewald}. 
Inspired by the “random mini-batch” idea \cite{jin2020random}, the random batch Ewald (RBE) method was proposed to further enhance Ewald-type approaches, enabling an ${\cal O}(N)$ Ewald method for efficient molecular dynamics simulations \cite{jin2021Coulomb,jin2021Convergence,jin2022Coulomb,liang2022RBM,liang2023random,GAN2025Coulomb,huang2025RBM}.
To improve the efficiency of force calculations in short-range interactions, extensive researches has been conducted \cite{plimpton1995fast,cornwell2000parallel,karakasidis2005parallel,friedrichs2009GPU,pechlaner2021multiple}. 
A common strategy for reducing computational cost is to introduce a cutoff radius, approximating the interaction force as zero when the interatomic distance exceeds this threshold. 
Neighbor lists, such as the Verlet-style neighbor list \cite{frenkel2023MD,spreiter1999Verlet}, are widely used to store neighboring particles within this radius. Due to its advantages in parallel efficiency and memory usage, the Verlet-style neighbor list method has become a standard technique in widely used MD software, such as LAMMPS \cite{THOMPSON2022LAMMPS} and GROMACS \cite{valdes2021GROMACS}.
The prefactor of the linear-scaling neighbor list algorithm depends on the average number of particles within the cutoff radius, which can be large in heterogeneous systems due to the need for a larger radius.


To further enhance the Verlet-style neighbor list method,  inspired by the random batch method (RBM) introduced by Jin et al.\thinspace\cite{jin2020random}, Xu and Liang et al.\thinspace\cite{liang2021RBL,liang2023random,liang2024NVE,gao2024RBMD} proposed the random batch list (RBL) method for calculating short-range interactions in MD simulations. This method decomposes the Verlet neighbor list into two separate neighbor lists based on the “core-shell” structure surrounding each particle. The random batch concept is then applied to construct a mini-batch of particles randomly selected from the shell zone, thereby reducing the computational cost of calculating interactions in the this zone. Current research efforts are primarily focused on applying the RBL method to the liquid systems with LJ potential. 
To the best of our knowledge, no studies have explored the application of the RBL method to other long-range potentials, such as the EAM potential.

The  EAM potential is a semi-empirical, many-body potential used to compute the total energy of a metallic system and has been successfully applied to bulk and interface problems \cite{daw1985application,foiles1988application,foiles1989application}.
It is derived from the second-moment approximation to tight-binding theory, which models the energy of the metal as the energy obtained by embedding an atom into the host electron density provided by the surrounding atoms. In the EAM potential, interactions between atoms are decomposed into pairwise interactions and atom-host interaction. The pairwise interactions share a structure similar to that of the LJ potential, while the atom-host interaction is inherently more complex than the simple pair-bond model.  As a result, the embedding function accounts for important many-atom interactions \cite{daw1993embedded}. In contrast to the pairwise LJ potential, the many-body EAM potential requires computing the host electron density, which involves a summation over neighboring atoms. This summation makes it challenging to directly generalize the RBL method to the EAM potential.

In this work, we extend the RBL method to metallic systems with the EAM potential, resulting in the Random Batch List for EAM potential (RBL-EAM) method. Similar to the RBL method, the proposed RBL-EAM constructs a "core-shell" list to efficiently approximate the host electron density. Both pairwise interactions and the embedding force within the EAM potential are then approximately computed using a separate "core-shell" list. Drawing inspiration from the “random mini-batch” idea \cite{jin2020random}, the batch size can be much smaller than the number of particles in the shell zone, leading to significant reductions in computational complexity and storage requirements by several folds. 
Under the assumption that the embedding energy is a quadratic function of the host electron density, we provide a theoretical, unbiased estimate of the host electron density and the force acting on each particle. Since the particle system with the EAM potential is a Newton-pair system, the computational complexity of force terms is halved by utilizing a half neighbor list. To further enhance efficiency, we combine the RBL-EAM with the half neighbor list and present a theoretical, unbiased estimate of the force in this case. The accuracy and efficiency of the RBL-EAM method are rigorously validated through several benchmark simulations.


The remainder of this paper is organized as follows. In \autoref{sec:2}, we briefly introduce the MD simulation for metallic systems with the EAM potential. \autoref{sec:3} presents a detailed description of the newly proposed RBL-EAM method and the theoretical unbiased estimate of the force terms. Numerical examples, including the lattice constant, radial distribution function, and elastic constants, are presented in \autoref{sec:4} to validate the accuracy and effectiveness of the RBL-EAM method. Finally, we draw our conclusions in \autoref{sec:5}.

\section{Molecular dynamic simulation for metallic system with EAM potential}\label{sec:2}
We consider a metallic system composed of $N$ atoms. The $i$-th atom, located at position $\bm{q}_i$, $i = 1, 2, \cdots$, $N$, is confined within a rectangular box with side length $L_1 \times L_2 \times L_3$.
The total energy of this system can be defined as a summation of the following form:
\begin{equation}\label{eq:potential_all}
    \mathcal{U}(\bm{q}_1, \bm{q}_2, \cdots \bm{q}_{N})
       = \sum_{i=1}^{N} \mathcal{U}_i, 
\end{equation}
where $\mathcal{U}_i$ represents the energy associated with the $i$-th atom, determined by its interactions with the surrounding atoms. 
For metallic systems, the Embedded Atom Method (EAM) potential \cite{daw1984embedded} is a widely adopted interaction potential to construct $ \mathcal{U}_i$, which is a second moment approximation to the tight binding theory.
According to the EAM potential,
the energy $ \mathcal{U}_i$ of the $i$-th atom can be uniformly expressed as follows:
\begin{equation}\label{eq:potential_u_i}
  \mathcal{U}_i = \mathcal{F}(\bar{\rho}_i)
  +
  \frac{1}{2}\sum_{\substack{j(j\neq i)\\  |\bm{r}_{ij}| \leq r_s}} \mathcal{V} (|\bm{r}_{ij}|),
\end{equation}
where $|\bm{r}_{ij}|$ with $\bm{r}_{ij} = \bm{q}_i - \bm{q}_j$ denotes the distance between the $i$-th and $j$-th atoms,  
$\mathcal{V}(|\bm{r}_{ij}|)$ is a pairwise potential energy as a function of $|\bm{r}_{ij}|$, and $r_s$ is the cutoff radius, respectively. 
Atoms within the cutoff radius $|\bm{r}_{ij}| \leq r_s$ are denoted as neighbors of the $i$-th atom.
We represent the set of all neighbors of atom $i$ as $\mathfrak{R}_i = \bigl\{j\big |\,|\bm{r}_{ij}| \leq r_s\bigr\}$.
In \eqref{eq:potential_u_i}, $\mathcal{F}(\bar{\rho}_i)$ is the embedding energy that represents the energy required to place the $i$-th atom into the electron cloud.
The host electron density $\bar{\rho}_i$ for atom $i$ is closely approximated by a sum of the atomic density $\rho(|\bm{r}_{ij}|)$ contributed by its neighbor atoms $j$: 
\begin{equation}\label{eq:bar_rho_i}
  \bar{\rho}_i = \sum_{\substack{j(j\neq i)\\  |\bm{r}_{ij}| \leq r_s}} \rho(|\bm{r}_{ij}|).
\end{equation}
%
The EAM potential defined in \eqref{eq:potential_u_i} requires the specification of three scalar functions:
the embedding function  $\mathcal{F}(\cdot)$, a pair-wise interaction $\mathcal{V}(\cdot)$, and an electron cloud contribution function $\rho(\cdot)$.
These functions are typically treated as fitting functions with undetermined parameters. According to \cite{daw1984embedded}, the function ${\cal F}$ should be fitted under several constrains, i.e., ${\cal F}$ has a single minimum and is linear at higher densities. 
The pairwise function ${\cal V}$ usually takes the following particular form: 
\begin{equation}
\label{pairwiseform}
   {\cal V}(|\bm{r}_{ij}|)=\frac{{\cal Z}_i(|\bm{r}_{ij}|){\cal Z}_j(|\bm{r}_{ij}|)}{|\bm{r}_{ij}|},
\end{equation}
where ${\cal Z}_i(|\bm{r}_{ij}|)$ is the effective charge of the $i$-th atom. 
To ensure the zero of energy correspond to neutral atoms separated to infinity, ${\cal Z}_i(|\bm{r}_{ij}|)$ should be monotonic and vanish continuously beyond a certain distance \cite{daw1984embedded}.
The electron density function $\rho(|\bm{r}|)$ \cite{banerjea1988origins} is defined as 
\begin{equation}
\label{exp-rho}
\rho(|\bm{r}|)= \rho_e\exp(-a(|\bm{r}|-r_e)),
\end{equation}
where $\rho_e>0$ is a scaling constant, $r_e$ is the equilibrium nearest distance, and $a>0$ is an adjustable parameter. For systems without external electric fields, the equilibrium nearest distance $r_e$ is typically set to zero.
For further details on the EAM potential, including interatomic potentials for various metals, we refer readers to  \cite{becker2013eam,hale2018eam,web_eam}.

For a system with $N$ atoms, the Lagrangian is defined as ${\cal L}=\sum_{i=1}^N \frac{  \bm{p}_i^2}{2m_i}-\mathcal{U}(\bm{q}_1, \bm{q}_2, \cdots \bm{q}_{N})$,
where $\bm{p}_i=m_i\dot{\bm{q}}_i$ represents the momenta of the $i$-th atom and $m_i$ is the mass. The Hamiltonian ${\cal H}$ of this system is given by ${\cal H}=\sum_{i=1}^N  \bm{p}_i\dot{\bm{q}}_i -{\cal L}$.  In this work, we assume that the system is in the canonical ensemble (NVT) \cite{nose1984NVT}. 
By introducing the Langevin thermostat,  the motion of the $N$ atoms is modeled as 
\begin{equation}\label{motion}
  \dot{\bm{q}}_i := \frac{\partial{\cal H}}{\partial \bm{p}_i}=\frac{\bm{p}_i}{m_i},\quad \dot{\bm{p}}_i := -\frac{\partial{\cal H}}{\partial \bm{q}_i}=-\frac{\partial \mathcal{U}(\bm{q}_1, \bm{q}_2, \cdots, \bm{q}_{N})}{\partial \bm{q}_i}-\gamma \bm{p}_i+\xi W_i,\quad i=1,\,2,\,\ldots,\,N.
\end{equation}
To ensure the system recovers the canonical ensemble distribution, the parameters $\gamma$ and $\xi$ are connected by a fluctuation-dissipation relation $\xi^2=2\gamma m_i k_B T$ with $k_B$ and $T$ being the Boltzmann constant and temperature, respectively. 

The MD simulation is performed by numerically solving the motion equations \eqref{motion} with a numerical time integration such as Verlet algorithm \cite{spreiter1999Verlet}.  
During the time integration, the primary computational cost arises from calculating the potential $\mathcal{U}(\bm{q}_1, \bm{q}_2, \cdots, \bm{q}_{N})$ and the force $\bm{\sigma}_i:=-{\partial \mathcal{U}(\bm{q}_1, \bm{q}_2, \cdots, \bm{q}_{N})}/{\partial \bm{q}_i}$. 
In EAM potential, the  potential is the summation of the embedding term and the pairwise term. Therefore, the force $\bm{\sigma}_i$ can be rewritten as
\begin{align}\label{eq:sigma_i}
  \begin{split}
    \bm{\sigma}_i  = -\frac{\partial \mathcal{U}(\bm{q}_1, \bm{q}_2, \cdots, \bm{q}_{N})}{\partial \bm{q}_i}
      := \bm{\sigma}^{e}_{i}  + \sum_{\substack{j(j\neq i)\\  |\bm{r}_{ij}| \leq r_s}} \bm{\sigma}^{p}_{ij},
  \end{split} 
\end{align}
where $\bm{\sigma}^{e}_{i}$ is the embedding force and $\bm{\sigma}^{p}_{ij}$ is the pairwise force between the $i$-th atom and the $j$-th atom. According to \eqref{eq:potential_all} and \eqref{eq:potential_u_i}, we derive 
\begin{align}\label{eq:sigma_i_components}
\bm{\sigma}^{e}_{i} & =-\sum_{j=1}^N\frac{\partial {\cal F}(\bar{\rho}_j)}{\partial \bm{q}_i}= - \sum_{\substack{j(j\neq i)\\  |\bm{r}_{ij}| \leq r_s}} 
\left(\mathcal{F}'(\bar{\rho}_i) {\rho}'(|\bm{r}_{ij}|)\frac{\partial |\bm{r}_{ij}|}{\partial \bm{q}_i}+ \mathcal{F}'(\bar{\rho}_j) {\rho}'(|\bm{r}_{ji}|)\frac{\partial |\bm{r}_{ji}|}{\partial \bm{q}_i}\right),\\
  \bm{\sigma}^{p}_{ij} &= -\frac12 \Bigl[
  {\mathcal{V}'(|\bm{r}_{ij}|)}\frac{\partial |\bm{r}_{ij}|}{\partial \bm{q}_i} +
  {\mathcal{V}'(|\bm{r}_{ji}|)}\frac{\partial |\bm{r}_{ji}|}{\partial \bm{q}_i}
  \Bigr]={\mathcal{V}'(|\bm{r}_{ij}|)}\frac{\partial |\bm{r}_{ij}|}{\partial \bm{q}_i}. \label{eq:sigma_i_components-p}
\end{align}
For the pairwise force $\bm{\sigma}^{p}_{ij}$, we have $\bm{\sigma}^{p}_{ij}=-\bm{\sigma}^{p}_{ji}$.
Let $\bm{\sigma}^{e}_{ij,1}=-\mathcal{F}'(\bar{\rho}_i) {\rho}'(|\bm{r}_{ij}|)\frac{\partial |\bm{r}_{ij}|}{\partial \bm{q}_i}$ and $\bm{\sigma}^{e}_{ij,2}=- \mathcal{F}'(\bar{\rho}_j) {\rho}'(|\bm{r}_{ji}|)\frac{\partial |\bm{r}_{ji}|}{\partial \bm{q}_i}$ be the partial embedding force contributed by host electron density $\bar{\rho}_i$ and $\bar{\rho}_j$, respectively. 
As there exist ${\mathcal F}'(\bar{\rho}_i)$ and ${\mathcal F}'(\bar{\rho}_j)$ in \eqref{eq:sigma_i_components} and $\bar{\rho}_i= \bar{\rho}_j$ is usually unavailable, the embedding force cannot be taken as a pairwise force.
However, we can still obtain the following force balance
\begin{equation}\label{forcebalance}
\begin{aligned}
\bm{\sigma}^{e}_{ij,1}+\bm{\sigma}^{e}_{ij,2}&=-\mathcal{F}'(\bar{\rho}_i) {\rho}'(|\bm{r}_{ij}|)\frac{\partial |\bm{r}_{ij}|}{\partial \bm{q}_i}- \mathcal{F}'(\bar{\rho}_j) {\rho}'(|\bm{r}_{ji}|)\frac{\partial |\bm{r}_{ji}|}{\partial \bm{q}_i}\\&=\mathcal{F}'(\bar{\rho}_i) {\rho}'(|\bm{r}_{ji}|)\frac{\partial |\bm{r}_{ji}|}{\partial \bm{q}_j}+\mathcal{F}'(\bar{\rho}_j) {\rho}'(|\bm{r}_{ji}|)\frac{\partial |\bm{r}_{ji}|}{\partial \bm{q}_j}=-\bm{\sigma}^{e}_{ji,2}-\bm{\sigma}^{e}_{ji,1}.
\end{aligned}    
\end{equation}

According to \eqref{exp-rho}, we observe that $\bm{\sigma}^{e}_{ij,1}$ and $\bm{\sigma}^{e}_{ij,2}$ decay exponentially as $|\bm{r}_{ij}|$  approaches infinity. Due to the form of ${\cal V}$ given in \eqref{pairwiseform}, the pairwise force  $\bm{\sigma}^{p}_{ij}$ also decays as $|\bm{r}_{ij}|$ goes to infinity, with the decay rate depending on the effective charge ${\cal Z}$. 
This is why a cutoff radius is introduced in the EAM potential. It is well known that the computational cost of potential and force for EAM increase rapidly as the cutoff radius $r_s$ increases.The cutoff radius $r_s$ depends on the lattice structure of metals. 
In this paper, we consider three types of lattice structures (see \autoref{fig:RVE_Structure}): body-centered cubic (BCC) lattice structure, face-centered cubic (FCC) lattice structure, and hexagonal close-packed (HCP) lattice structure. 
The EAM potential as a function of $|\bm{r}|$ is plotted for metal Cu (FCC), $\alpha$-Fe (BCC), and Mg (HCP) in \autoref{fig:neighbor_FCC-BCC-HCP}.
Due to the lattice structure, the neighbor set $\mathfrak{R}_i$ can be decomposed into the $1$-st, $2$-nd, ..., and $7$-th nearest neighbors. By assuming that all atoms are in the equilibrium positions, we count the number of atoms in the $i$-th near neighbor and calculate the contributions of these atoms to the total potential energy. The results, shown in \autoref{fig:neighbor_FCC-BCC-HCP}, indicate that the first two nearest neighbors account for less than $25\%$ ($12.8\%$ for FCC lattice, $15.9\%$ for BCC lattice and $22.5\%$ for the HCP lattice) of the atoms but contribute more than $70\%$ of the energy. However, as the numerical results demonstrated in \autoref{sec:4}, considering only the first and second nearest neighbors in the neighbor set  $\mathfrak{R}_i$, may lead to non-physical simulation results. 
\begin{figure}[htbp]
  \centering
  \includegraphics[width=0.90\textwidth]{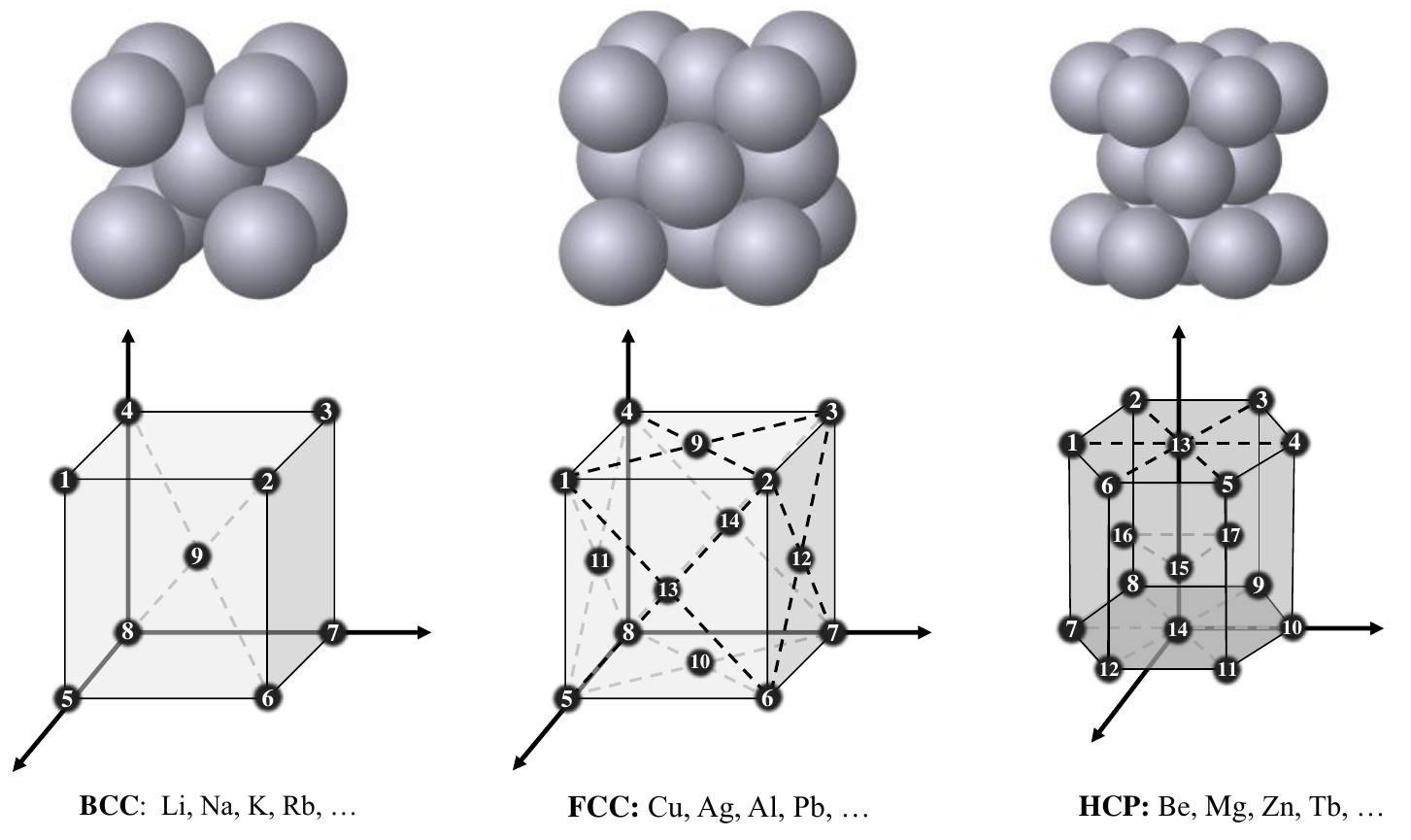}
  \caption{The lattice structures for the metal crystals.}\label{fig:RVE_Structure}
\end{figure} 

In recent years, the random batch method (RBM) proposed by Jin et al.\thinspace\cite{jin2020random,jin2021Convergence,jin2021Coulomb} aims to accelerate the calculation of the potential and force terms in MD simulation. 
Initially applied to long-range interactions (e.g., electrostatic interactions), RBM has been extended to short-range interactions \cite{liang2021RBL,gao2024RBMD,liang2024NVE} for the Lennard-Jones (LJ) potential.
Given the widespread use of the EAM potential in metallic materials,  extending RBM to EAM is of significant interest. 
However, unlike the LJ pairwise potential, the EAM potential includes an embedding term, which introduces additional complexity due to its non-pairwise nature. 
This additional truncation complicates the implementation of RBM and makes the estimation of numerical errors more challenging. 
These challenges motivate the proposal of the RBL-EAM method in the next section.

%
\begin{figure}[htbp]
  \centering
  \subfloat[Cu(FCC)]{\includegraphics[width = 0.65\textwidth]{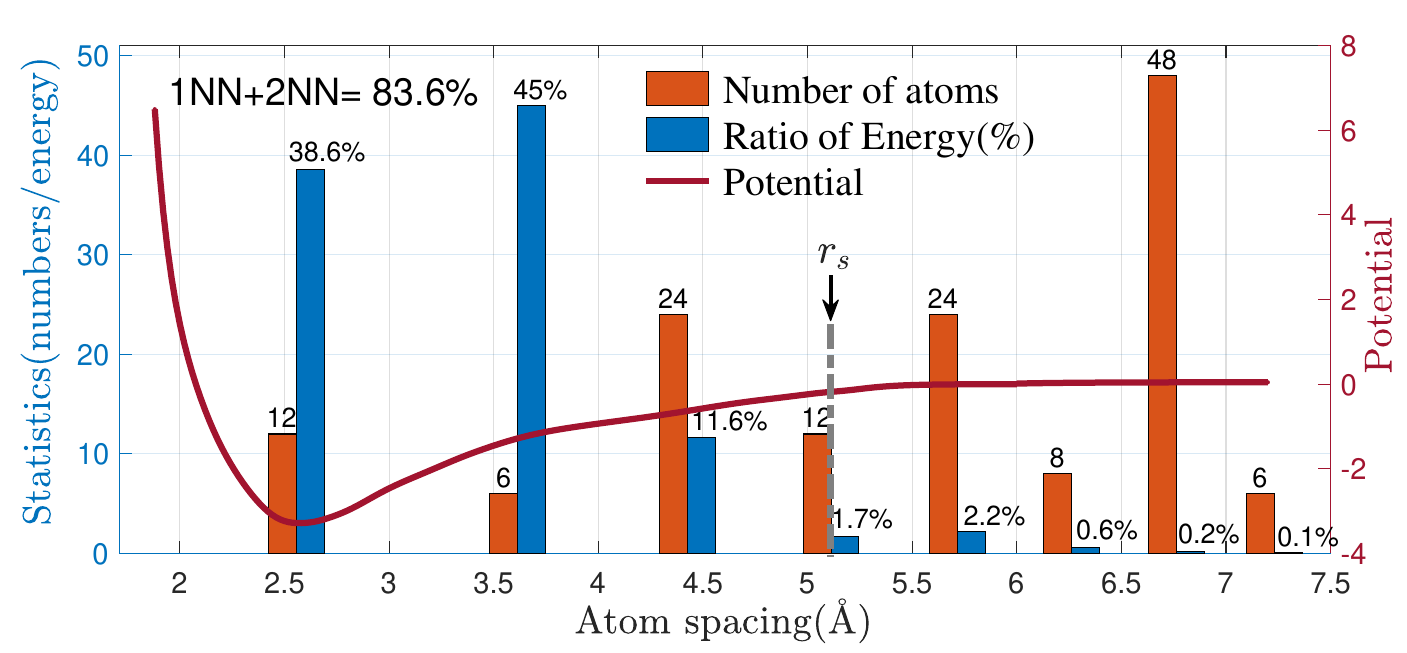}}\\
  \subfloat[Fe(BCC)]{\includegraphics[width = 0.65\textwidth]{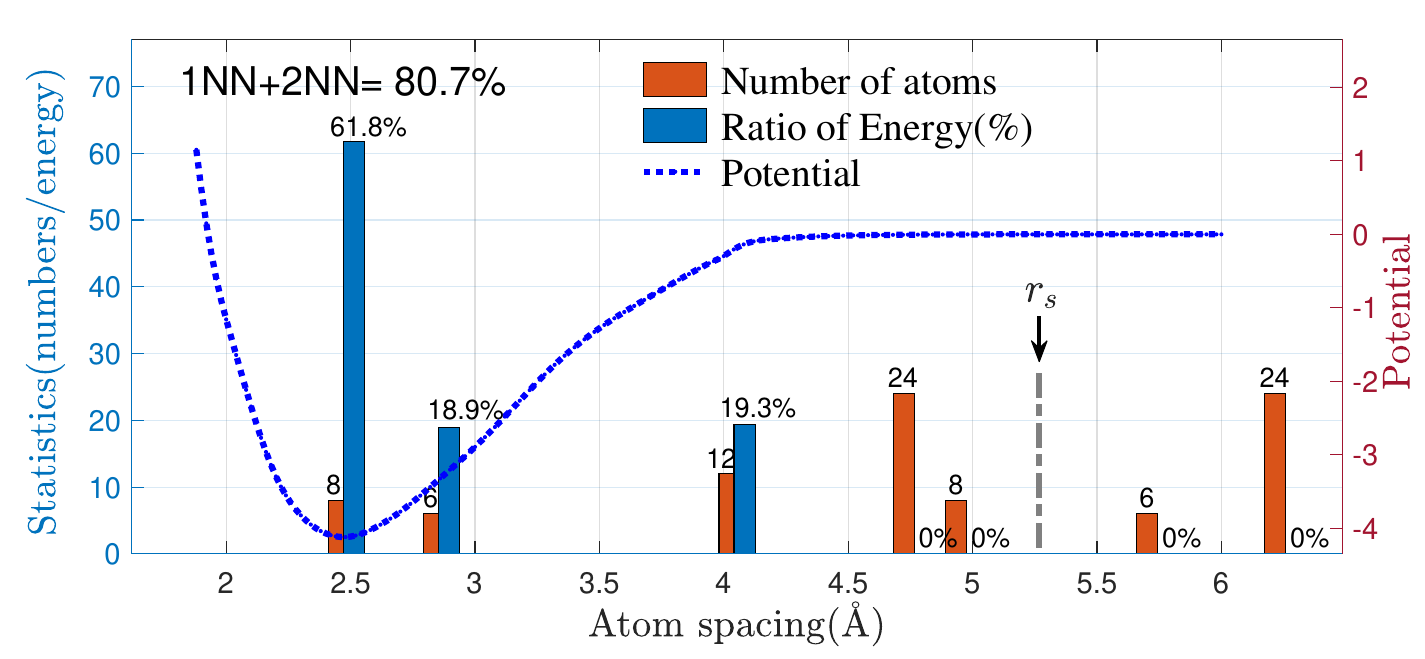}}\\
  \subfloat[Mg(HCP)]{\includegraphics[width = 0.65\textwidth]{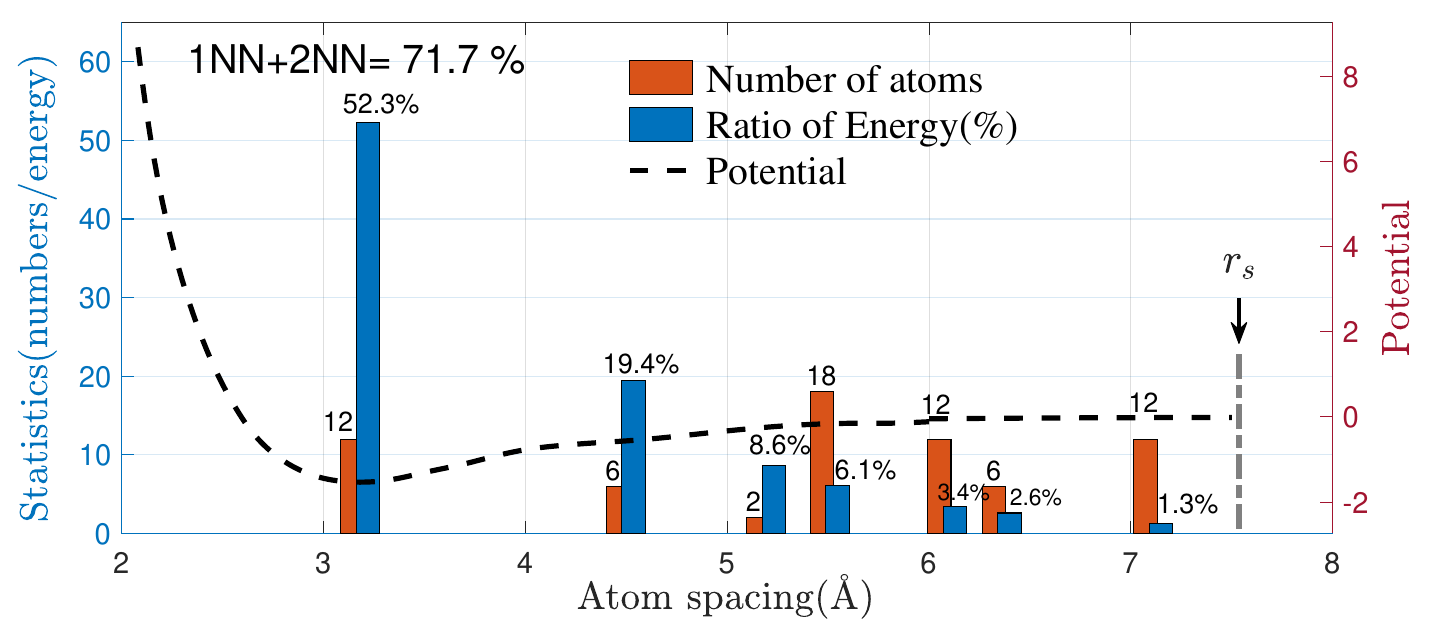}}
  \caption{EAM potential, the number of atoms, and the contributions to the total potential energy for each nearest neighbor. (a) Metal Cu (FCC), (b) Metal $\alpha$-Fe (BCC), and (c) Metal Mg (HCP).}\label{fig:neighbor_FCC-BCC-HCP}
\end{figure}

\section{The random batch list method for EAM potential}\label{sec:3}
Inspired by \cite{liang2021RBL}, we extend the random batch method to the metallic systems with EAM potential, resulting the RBL-EAM method. 
This newly presented method consist of two main components: the calculation of host electron density $\bar{\rho}_i$ and the computation of force terms.
%

\subsection{The random batch list method for host electron density}\label{sub:electron}
Following the idea of RBL method \cite{liang2021RBL}, we introduce a core-shell structured neighbor list for each atom, which is derived from the stenciled cell list and implements on the top of LAMMPS. 
As shown in \autoref{fig:rbl_box}, a small spherical core with radius $|\bm{r}_{ij}|\leq r_c$ and a shell with $r_c< |\bm{r}_{ij}|\leq r_s+\Delta r$ are constructed around each atom. Here $\Delta r$ is a “skin” distance that allows the neighbor list to be reused for multiple steps.  
Atoms within the core radius are treated using direct summation, while atoms in the shell are handled using a random batch approach. This strategy reduces the number of interacting pairs while maintaining computational accuracy.
In the neighbor list of atom $i$, atoms are classified into two types: 
(1) I-type atoms within the small core defined as $\mathfrak{R}_i^{c} = \bigl\{j\big| |\bm{r}_{ij}|\leq r_c \bigr\}$;
and (2) II-type atoms outside the core but within the shell, defined as $\mathfrak{R}_i^{s} = \bigl\{j\big| r_c<|\bm{r}_{ij}|\leq r_s +\Delta r \bigr\}$. 
For further details on the core-shell structured neighbor list, readers are referred to \cite{liang2021RBL}. 
\begin{figure}[htbp]
  \centering
  \subfloat[2-dimensional case]{\includegraphics[width = 0.42\textwidth]{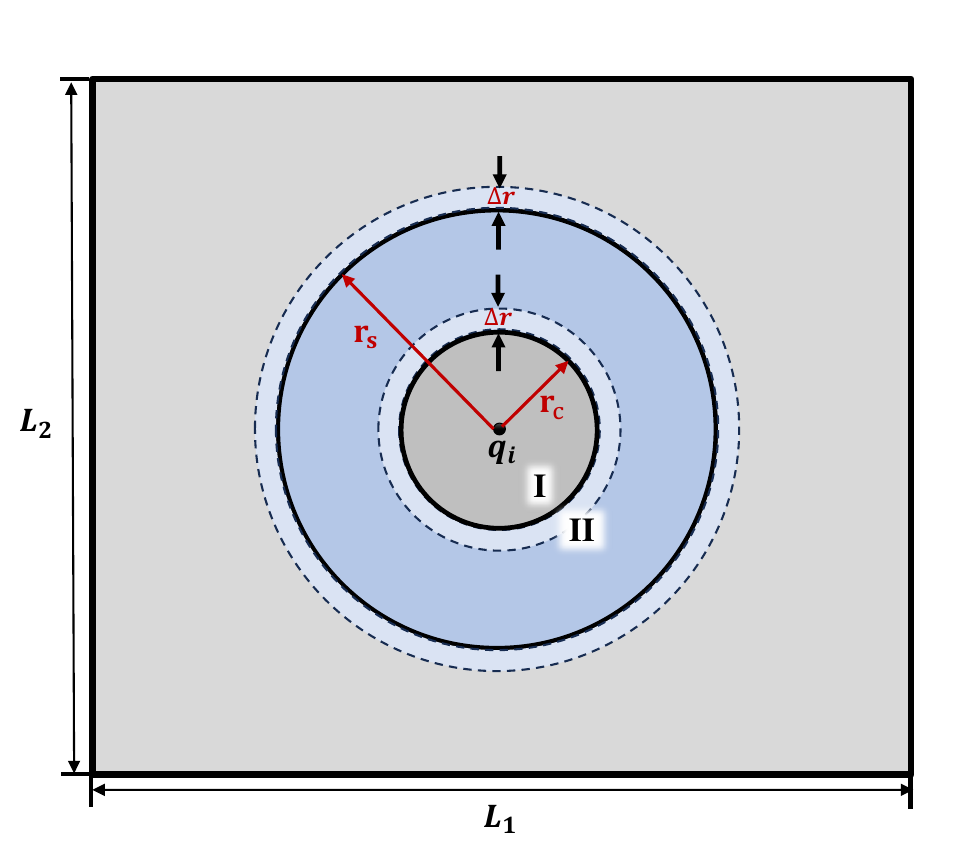}}
  \subfloat[3-dimensional case]{\includegraphics[width = 0.37\textwidth]{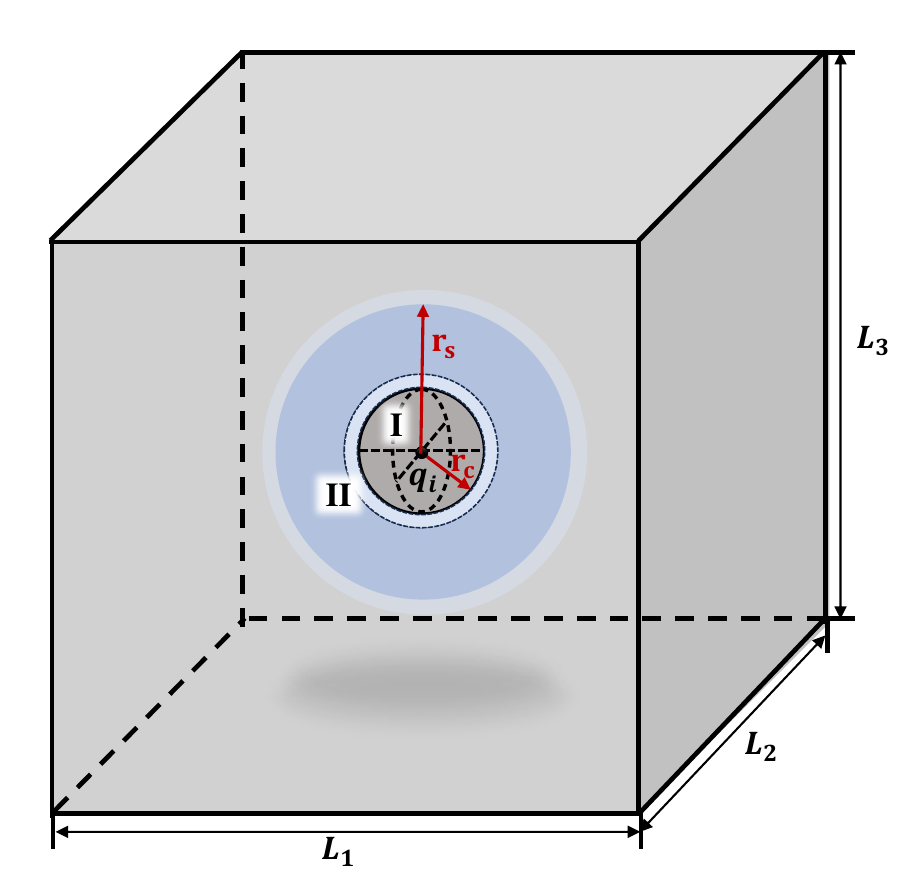}}
  \caption{The core-shell structured neighbor list for the $i$-th atom. (a) 2-dimensional case and (b) 3-dimensional case. Here $r_c$ is the cutoff radius of core region, $r_s$ is the cutoff radius for the shell region, and $\Delta r$ is a “skin” distance that set according to the format of neighbor list in LAMMPS.}\label{fig:rbl_box}
\end{figure}

We then develop a random method to efficiently calculate the host electron density $\tilde{\bar{\rho}}_i:=  \bar{\rho}_i^c + \tilde{\bar{\rho}}_i^s$, which is expressed as the summation of two parts from I-type atoms and II-type atoms. The contribution to the summation of atomic electron density $\rho(|\bm{r}_{ij}|)$ from I-type atoms is calculated directly as
\begin{align} \label{eq:rho_c}
  \bar{\rho}_i^c = \sum_{\substack{j(j\neq i)\\  |\bm{r}_{ij}| \leq r_c}} \rho(|\bm{r}_{ij}|).
\end{align}
For the II-type atoms, we employ a random batch approach. 
Given a batch size $p$, we randomly select $p$ atoms from all II-type atoms. Let $N_i^s$ denote the number of II-type atoms. 
If $N_i^s\leq p$, the host electron density $\tilde{\bar{\rho}}_i^s$ is calculated by considering all II-type atoms, resulting in $\tilde{\bar{\rho}}_i={\bar{\rho}}_i$. 
Otherwise, we select $p$ atoms from II-type atoms $\mathfrak{R}_i^{s}$ to form the batch ${\mathfrak B}_i$, and approximate the host electron density contributed by the II-type atoms as follows: 
\begin{align} \label{eq:rho_s}
  \tilde{\bar{\rho}}_i^s = \frac{N_i^s}{p} \sum_{j \in \mathfrak{B}_i} \rho(|\bm{r}_{ij}|),
\end{align}
where $\rho(|\bm{r}_{ij}|)=0$ with $|\bm{r}_{ij}|>r_s$.
The calculation process for the host electron density $\tilde{\bar{\rho}}_i$ using the RBL-EAM method is outlined in Algorithm \ref{Algorithm_1}. 

\begin{algorithm}[htbp] 
\footnotesize{
  \caption{RBL-EAM: the calculation of electron density $\tilde{\bar{\rho}}_i$.}\label{Algorithm_1}
  \SetAlgoLined
  \begin{algorithmic}[1]
   \STATE Divide the neighboring list of atom $i$ into $\mathfrak{R}_i^{c}$ and $\mathfrak{R}_i^{s}$.
    \STATE  For all {$j \in \mathfrak{R}_i^{c}$}, calculate the electron density $\bar{\rho}_i^c$ using \eqref{eq:rho_c}.
   \STATE  Select $p$ atoms randomly from $\mathfrak{R}_i^{s} $ into a batch $\mathfrak{B}_i$.
     \STATE  For all {$j \in \mathfrak{B}_i$}, compute the electron density $\tilde{\bar{\rho}}_i^s$ according to \eqref{eq:rho_s}.
   \end{algorithmic}
   \KwOut{The host electron density $\tilde{\bar{\rho}}_i =  \bar{\rho}_i^c + \tilde{\bar{\rho}}_i^s$.} 
}
\end{algorithm}

Let $\mathcal{Q}_i = \tilde{\bar{\rho}}_i - \bar{\rho}_i$. 
The following \autoref{theorem-1} demonstrates that the expectation of $\mathcal{Q}_i$ is zero and its variance is bounded, validating the unbiased nature of our approximation for the host electron density.
%
\begin{theorem}\label{theorem-1}
 The expectation of  $\mathcal{Q}_i$ is zero and its variance is bounded.
\end{theorem}
\begin{proof}
We rewrite the host electron density ${\bar{\rho}}_i$ as the summation of contributions from the core and shell regions, i.e., ${\bar{\rho}}_i =  \bar{\rho}_i^c + {\bar{\rho}}_i^s$ with
\begin{align} \label{eq:rho_bar_s}
{\bar{\rho}}_i^s = \sum_{\substack{j(j\neq i)\\   r_c < |\bm{r}_{ij}| \leq r_s}} \rho(|\bm{r}_{ij}|).
\end{align}
To facilitate the analysis, let us introduce an indicative function $H_{ij}$, which is equal to $1$ for $j \in \mathfrak{B}_i$ and $0$ otherwise. 
The expectation of $\tilde{\bar{\rho}}_i^s$ is calculated as 
\begin{equation}\label{eq:Expec-rho}
  \mathbb{E} [\tilde{\bar{\rho}}_i^s] = \frac{N_i^s}{p} \sum_{j\in\mathfrak{R}^s_i} \rho(|\bm{r}_{ij}|) \mathbb{E} [H_{ij}] = \bar{\rho}_i^s,
\end{equation} 
which implies: 
\begin{equation}\label{eq:Expec-Q}
    \mathbb{E}[\mathcal{Q}_i]=\mathbb{E}[\tilde{\bar{\rho}}_i] -{\bar{\rho}}_i =\bar{\rho}_i^c+\mathbb{E} [\tilde{\bar{\rho}}_i^s] -(\bar{\rho}_i^c+\bar{\rho}_i^s)= 0.
\end{equation}
The variance of the quantity $\mathcal{Q}_i$ is given by 
\begin{equation}\label{eq:var_rho}
  \begin{split}
    \mathbb{V}[{\mathcal Q}_i] 
     = \mathbb{V}[\tilde{\bar{\rho}}^s_i - {\bar{\rho}}^s_i]  = \mathbb{E}\bigl[ (\tilde{\bar{\rho}}_i^s)^2\bigr] - ({\bar{\rho}}_i^s)^2. 
  \end{split}
\end{equation}
Using the indicative function $H_{ij}$, we can derive
\begin{align}\label{eq:17}
  \begin{split}
    \mathbb{E}\bigl[ (\tilde{\bar{\rho}}_i^s)^2\bigr] &= \frac{(N_i^s)^2}{p^2} \mathbb{E}\biggl[
     \Bigl( \sum_{r_c < |\bm{r}_{ij}| \leq r_s} \rho(|\bm{r}_{ij}|) H_{ij}\Bigr)^2
    \biggr] \\
    & = \frac{(N_i^s)^2}{p^2} 
    \mathbb{E}\biggl[\sum_{r_c < |\bm{r}_{ij}| \leq r_s} \rho^2(|\bm{r}_{ij}|) H_{ij}
      + \sum_{r_c < |\bm{r}_{ij}| \leq r_s} \sum_{\substack{ k(k\neq j)\\  r_c < |\bm{r}_{ik}| \leq r_s}} \rho(|\bm{r}_{ij}|)\rho(|\bm{r}_{ik}|) H_{ij} H_{ik}\biggr]  \\
    & = \frac{N_i^s}{p} \biggl[\sum_{r_c < |\bm{r}_{ij}| \leq r_s} \rho^2(|\bm{r}_{ij}|) + \frac{p-1}{N_i^s-1} 
    \sum_{r_c < |\bm{r}_{ij}| \leq r_s} \sum_{\substack{ k(k\neq j)\\  r_c < |\bm{r}_{ik}| \leq r_s}}\rho(|\bm{r}_{ij}|)\rho(|\bm{r}_{ik}|) \biggr]\\
    &\leq \frac{N_i^s}{p} \biggl[\sum_{r_c < |\bm{r}_{ij}| \leq r_s} \rho^2(|\bm{r}_{ij}|) +  
    \sum_{r_c < |\bm{r}_{ij}| \leq r_s} \sum_{\substack{ k(k\neq j)\\  r_c < |\bm{r}_{ik}| \leq r_s}}\rho(|\bm{r}_{ij}|)\rho(|\bm{r}_{ik}|) \biggr]=\frac{N_i^s}{p} ({\bar{\rho}}_i^s)^2, 
  \end{split}
  \end{align}
 where the last inequality holds due to $\rho(|\bm{r}_{ij}|)\geq 0$. 
%
Therefore, the variance of  $\mathcal{Q}_i$ is bounded as follows:
\begin{align}\label{eq:val_inequality}
    \mathbb{V}[{\mathcal Q}_i] 
    \leq \frac{N_i^s-p}{p}  ({\bar{\rho}}_i^s)^2. 
\end{align} 
For a homogeneous system with atomic density $\varrho$, the leading term of $\mathbb{V}[{\mathcal Q}_i]$ has the following asymptotic bound:
\begin{equation}
    \begin{split}
    \mathbb{V}[\mathcal{Q}_i] 
    & \leq \frac{N_i^s-p}{p} \left(\sum_{j\in{\mathfrak{R}_i}} \rho(|\bm{r}_{ij}|)\right)^2 \\
    & \lesssim \frac{N_i^s-p}{p} \left(\int_{r_c}^{r_s} 4\pi \varrho r^2\cdot \rho_e\exp(-ar)  \textnormal{d}r\right)^2 \\
    & \leq \frac{N_i^s-p}{p} \left[\frac{4\pi \varrho\rho_e}{a^3}(a^2 r_c^2 + 2a r_c+2)\exp(-ar_c) \right]^2.
    \end{split}
\end{equation}
This shows that the variance of ${\mathcal Q}_i$ decays rapidly as $r_c$ increase.
\end{proof}
\begin{figure}[htbp]
    \centering
    \includegraphics[width = 0.75\textwidth]{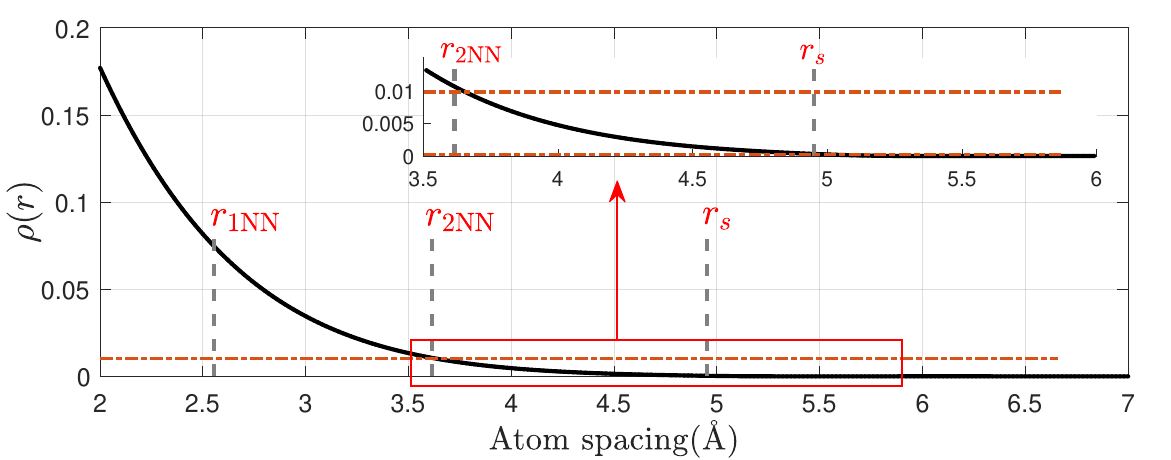}
    \caption{The electron density function $\rho(r)$ for metal Cu.}
    \label{fig:Rho_r}
\end{figure}
\begin{remark}\label{remark-1}
As suggested by \cite{banerjea1988origins}, the electron density function $\rho(r)$ decays exponentially. In \autoref{fig:Rho_r}, we plot the curve of electron density function $\rho(r)$ for metal Cu, which clearly shows that $\int_{r_c}^{r_s}\rho(r)\textnormal{d} r$ is bounded. 
Moreover, if the core radius $r_c$ is chosen such that all $1$-st and $2$-nd nearest neighbor atomic distance are included (i.e., $r_c>r_{\rm{2NN}}$), then $\int_{r_c}^{r_s}\rho(r)\textnormal{d} r$ is close to zero. In this case, the variance of ${\mathcal Q}_i$ is also close to zero. 
\end{remark}

\subsection{The random batch list method for the force terms}
\label{forcefullsystem}
Next, we consider the RBL approximation of the interaction force defined in Eq.\eqref{eq:sigma_i}. 
Similar to the host electron density $\tilde{\bar{\rho}}_i$, the approximate interaction force is expressed as:
\begin{equation}\label{eq:total_sigma}
  \tilde{\bm{\sigma}}_i = \tilde{\bm{\sigma}}_i^{e} + {\bm{\sigma}}_i^{pc} + \tilde{\bm{\sigma}}_i^{ps}:=\tilde{\bm{\sigma}}_i^{ec}+\tilde{\bm{\sigma}}_i^{es}+ {\bm{\sigma}}_i^{pc} + \tilde{\bm{\sigma}}_i^{ps},
\end{equation}
where 
\begin{subequations}
  \begin{align}
  \tilde{\bm{\sigma}}^{e}_{i} & 
  =\tilde{\bm{\sigma}}_i^{ec}+\tilde{\bm{\sigma}}_i^{es}, \\
   \tilde{\bm{\sigma}}_i^{ec} & = \sum_{\substack{j(j\neq i)\\|\bm{r}_{ij}| \leq r_c}} \tilde{\bm{\sigma}}_{ij}^{e}= \sum_{\substack{j(j\neq i)\\|\bm{r}_{ij}| \leq r_c}} 
\left(\tilde{\bm{\sigma}}^{e}_{ij,1}+ \tilde{\bm{\sigma}}^{e}_{ij,2}\right),\label{eq:sigma_ec}\\ 
    \tilde{\bm{\sigma}}_i^{es} & = \frac{N_i^s}{p}\sum_{j\in\mathfrak{B}'_i}\tilde{\bm{\sigma}}_{ij}^{e}=  \frac{N_i^s}{p}\left(
     \sum_{j\in\mathfrak{B}'_i} 
\tilde{\bm{\sigma}}^{e}_{ij,1}+ \sum_{ j\in\mathfrak{B}'_i} \tilde{\bm{\sigma}}^{e}_{ij,2}\right)
, \label{eq:sigma_es}\\ 
    \bm{\sigma}_i^{pc} & = \sum_{\substack{j(j\neq i)\\  |\bm{r}_{ij}| \leq r_c}} \bm{\sigma}^{p}_{ij}, \label{eq:sigma_pc}
    \\
    \tilde{\bm{\sigma}}_i^{ps} & = \frac{N_i^s}{p} \sum_{j \in \mathfrak{B}'_i}  \bm{\sigma}^{p}_{ij}. \label{eq:sigma_ps}
  \end{align}
\end{subequations}
Here $\tilde{\bm{\sigma}}^{e}_{ij,1}=-\mathcal{F}'(\tilde{\bar{\rho}}_i) {\rho}'(|\bm{r}_{ij}|)\frac{\partial |\bm{r}_{ij}|}{\partial \bm{q}_i}$ and $\tilde{\bm{\sigma}}^{e}_{ij,2}=- \mathcal{F}'(\tilde{\bar{\rho}}_j) {\rho}'(|\bm{r}_{ji}|)\frac{\partial |\bm{r}_{ji}|}{\partial \bm{q}_i}$. 
In \eqref{eq:total_sigma}, $\tilde{\bm{\sigma}}_i^{ec}$ and ${\bm{\sigma}}_i^{pc}$ are the embedding force and pairwise force in the core region, computed directly. Since the host electron density $\bar{\rho}_i$ and $\bar{\rho}_j$ are approximately by the RBL-EAM method in \autoref{sub:electron}, $\tilde{\bm{\sigma}}^{e}_{ij,k}$ (with $k=1,2$) differ from ${\bm{\sigma}}^{e}_{ij,k}$. 
For a given batch size $p$, we randomly select $p$ atoms from all II-type atoms, independent of the random selection used in the calculation of $\tilde{\bar{\rho}}_i$. 
This yields a new random batch set $\mathfrak{B}_i'$, which is used to construct $\tilde{\bm{\sigma}}_i^{es}$ as defined in \eqref{eq:sigma_es}.
Algorithm \ref{Algorithm_2} outlines the calculation process for the interaction force $\tilde{\bm{\sigma}}_i$ using the newly proposed RBL-EAM method.
After calculating the net force for all atoms, we subtract the average net force $\bm{\sigma}_{ave}$ on each particle to ensure zero net force, 
thereby conserving total momentum.
%

\begin{algorithm}[htbp] 
\footnotesize{
  \caption{RBL-EAM: the calculation of force terms $\tilde{\bm{\sigma}}_i$. }\label{Algorithm_2}
  \SetAlgoLined
  \KwIn{$\tilde{\bar{\rho}}_i$ obtained by Algorithm \ref{Algorithm_1} and neighbor lists $\mathfrak{R}_i^{c}$ and $\mathfrak{R}_i^{s}$.}
  \begin{algorithmic}[1]
   \STATE For all $j \in \mathfrak{R}_i^{c}$, calculate the embedding force $\tilde{\bm{\sigma}}_i^{ec}$  and pairwise force $\bm{\sigma}_i^{pc}$ using \eqref{eq:sigma_ec} and \eqref{eq:sigma_pc}.
   \STATE  Select $p$ atoms randomly from $\mathfrak{R}_i^{s}$ into a batch $\mathfrak{B}'_i$.
   \STATE For all $j \in \mathfrak{B}'_i$, compute the embedding force $\tilde{\bm{\sigma}}_i^{es}$ and pairwise force $\tilde{\bm{\sigma}}_i^{ps}$ using \eqref{eq:sigma_es} and \eqref{eq:sigma_ps}.
   \end{algorithmic}
   \KwOut{The interaction force $\tilde{\bm{\sigma}}_i = \tilde{\bm{\sigma}}_i^{ec} + \tilde{\bm{\sigma}}_i^{es}+ {\bm{\sigma}}_i^{pc} + \tilde{\bm{\sigma}}_i^{ps}$.} 
}
\end{algorithm}

Let $\bm{\mathcal{X}}_i = \tilde{\bm\sigma}_i - \bm{\sigma}_i+\bm{\sigma}_{ave}$, where $\bm{\sigma}_i$ is the interacting force defined in \eqref{eq:sigma_i} with cutoff radius $r_s$. To demonstrate the accuracy of the approximation for interaction force in \eqref{eq:total_sigma}, we needs to prove that expectation of $\bm{\mathcal{X}}_i$ is zero and its variance is bounded. 
Following the idea of RBL \cite{liang2021RBL}, we can estimate the contribution of pairwise force to expectation and variance of $\bm{\mathcal{X}}_i$. 
However, as an approximate host electron density $\tilde{\bar{\rho}}_i$ is used to construct the embedding force $\tilde{\bm{\sigma}}_i^{ec}$ and $\tilde{\bm{\sigma}}_i^{es}$, the calculation of $\mathbb{E}[\tilde{\bm{\sigma}}_i^{ec}+\tilde{\bm{\sigma}}_i^{es}-{\bm{\sigma}}_i^{e}]$ and its variance are more challenging than that for the pairwise force.  
For a linear function ${\cal F}'(\cdot)$, \autoref{theorem-2} proves that the force calculated by the RBL-EAM is unbiased (i.e., the expectation of $\bm{\mathcal{X}}_i$ is zero) and its variance is bounded. 

\begin{theorem}\label{theorem-2}
 The expectation of $\bm{\mathcal{X}}_i$ is zero for linear function ${\cal F}'(\cdot)$ and close to zero otherwise.  Moreover, the variance of $\bm{\mathcal{X}}_i$ is bounded.
\end{theorem}
\begin{proof}
We decompose $\bm{\sigma}_i$ similarly to \eqref{eq:total_sigma}, i.e., $\bm{\sigma}_i = {\bm{\sigma}}_i^{e} + {\bm{\sigma}}_i^{pc} + {\bm{\sigma}}_i^{ps}$. According to the definition of $\bm{\sigma}_{ave}$, we have 
\begin{equation}\label{eq:sigma_ave}
  \bm{\sigma}_{ave} = \frac{1}{N} \sum_{i=1}^{N} \bigl(\tilde{\bm{\sigma}}_i^{e} + {\bm{\sigma}}_i^{pc} + \tilde{\bm{\sigma}}_i^{ps}\bigr).
\end{equation} 
The expectation of $\bm{\sigma}_{ave}$ is zero due to the force balance for pairwise force and \eqref{forcebalance}.
Thus, the expectation of $\bm{\mathcal{X}}_i$ is
\begin{equation}
  \begin{split}
    \mathbb{E}[\bm{\mathcal{X}}_i] &= \mathbb{E}[\tilde{\bm\sigma}_i - \bm{\sigma}_i] + \mathbb{E}\bigl[{\bm{\sigma}}_{ave} \bigr]  = \mathbb{E}\bigl[ \tilde{\bm{\sigma}}_i^{e} -{\bm{\sigma}}_i^{e} \bigr]
    + \mathbb{E}\bigl[\tilde{\bm{\sigma}}_i^{ps} - {\bm{\sigma}}_i^{ps} \bigr].
  \end{split}
\end{equation} 
Similar to the analysis of electron density, by introducing the indicative function $H'_{ij}$ for $j\in\mathfrak{B}'_i$, we have 
\begin{align}
  \mathbb{E} [\tilde{\bm{\sigma}}_i^{ps} - {\bm{\sigma}}_i^{ps}] &= 
  \frac{N_i^s}{p} \sum_{\substack{j(j \neq i)\\  r_c < |\bm{r}_{ij}| \leq r_s}}  \bm{\sigma}^{p}_{ij} \mathbb{E} [H'_{ij}]  
  - \sum_{\substack{j(j \neq i)\\  r_c < |\bm{r}_{ij}| \leq r_s}}  \bm{\sigma}^{p}_{ij} = 0. \label{eq:expect_ps}
\end{align}
If ${\mathcal{F}'(\rho)}$ is a linear function of $\rho$, due to $\mathbb{E} [\tilde{\bar{\rho}}_i]=\bar{\rho}_i$, we can further obtain 
\begin{equation*}
\mathbb{E} \bigl[ {\mathcal{F}'(\tilde{\bar{\rho}}_i)}\bigr]
=  \mathcal{F}'\bigl(\mathbb{E} [\tilde{\bar{\rho}}_i] \bigr)
=  \mathcal{F}'\bigl(\bar{\rho}_i \bigr).     
\end{equation*}
For $|\bm{r}_{ij}|\leq r_c$, the expectation of $\tilde{\bm{\sigma}}^{e}_{ij,1}$ is 
\begin{equation}
\label{eq:expect_ec-1}
\mathbb{E} \bigl[\tilde{\bm{\sigma}}^{e}_{ij,1}\bigr]=-\mathbb{E} \bigl[\mathcal{F}'(\tilde{\bar{\rho}}_i) \bigr]{\rho}'(|\bm{r}_{ij}|)\frac{\partial |\bm{r}_{ij}|}{\partial \bm{q}_i}={\bm{\sigma}}^{e}_{ij,1}. 
\end{equation}
Similarly, we get 
$\mathbb{E} \bigl[\tilde{\bm{\sigma}}^{e}_{ij,2}\bigr] ={\bm{\sigma}}^{e}_{ij,2}$. For $r_c<|\bm{r}_{ij}|\leq r_s$, we have
\begin{equation} 
\label{eq:26}
\begin{aligned}\mathbb{E} \bigl[(\tilde{\bar{\rho}}_i &-{\bar{\rho}}_i) H'_{ij}\bigr]= \frac{N_i^s}{p} \mathbb{E}\biggl[
      \sum_{r_c < |\bm{r}_{ik}| \leq r_s} \rho(|\bm{r}_{ik}|) H'_{ij}H_{ik}
    \biggr]-\frac{p}{N_i^s}{\bar{\rho}}^s_i\\
&=\frac{p}{N_i^s}
      \sum_{r_c < |\bm{r}_{ik}| \leq r_s} \rho(|\bm{r}_{ik}|) -\frac{p}{N_i^s}{\bar{\rho}}^s_i=0,
\end{aligned}
\end{equation}
where $\mathbb{E}\left[ H'_{ij}H_{ik}\right]=\mathbb{E}\left[ H'_{ij}\right]\mathbb{E}\left[ H_{ik}\right]=\frac{p^2}{(N_i^s)^2}$ due to the independence between $H'_{ij}$ and $H_{ik}$. Since $\cal{F}'(\rho)$ is a linear function, we obtain 
\begin{equation}
\label{eq:expect_es-1}
\begin{aligned}\mathbb{E} \bigl[\tilde{\bm{\sigma}}^{e}_{ij,1} H'_{ij}\bigr]&=-\mathbb{E} \bigl[\mathcal{F}'(\tilde{\bar{\rho}}_i) H'_{ij}\bigr]{\rho}'(|\bm{r}_{ij}|)\frac{\partial |\bm{r}_{ij}|}{\partial \bm{q}_i}\\
&=-\mathbb{E} \bigl[\mathcal{F}'(\bar{\rho}_i)H'_{ij} + \mathcal{F}''(\bar{\rho}_i) \cdot (\tilde{\bar{\rho}}_i -{\bar{\rho}}_i)H'_{ij}\bigr]{\rho}'(|\bm{r}_{ij}|)\frac{\partial |\bm{r}_{ij}|}{\partial \bm{q}_i}=\frac{p}{N_i^s}{\bm{\sigma}}^{e}_{ij,1},
\end{aligned}
\end{equation}
and 
\begin{equation}  
\label{eq:expect_es-2}  
\begin{aligned}\mathbb{E} \bigl[\tilde{\bm{\sigma}}^{e}_{ij,2} H'_{ij}\bigr]&=-\mathbb{E} \bigl[\mathcal{F}'(\tilde{\bar{\rho}}_j) H'_{ij}\bigr]{\rho}'(|\bm{r}_{ji}|)\frac{\partial |\bm{r}_{ji}|}{\partial \bm{q}_i}\\
&=-\mathbb{E} \bigl[\mathcal{F}'(\bar{\rho}_j)H'_{ij} + \mathcal{F}''(\bar{\rho}_j) \cdot (\tilde{\bar{\rho}}_j -{\bar{\rho}}_j)H'_{ij}\bigr]{\rho}'(|\bm{r}_{ji}|)\frac{\partial |\bm{r}_{ji}|}{\partial \bm{q}_i}=\frac{p}{N_i^s}{\bm{\sigma}}^{e}_{ij,2}.
\end{aligned}
\end{equation}
By combining \eqref{eq:expect_ps}, \eqref{eq:expect_ec-1}, \eqref{eq:expect_es-1} and \eqref{eq:expect_es-2}, we have
\begin{equation}
  \begin{split}
    \mathbb{E}[\bm{\mathcal{X}}_i] &=  \mathbb{E}\bigl[ \tilde{\bm{\sigma}}_i^{e} -{\bm{\sigma}}_i^{e} \bigr]\\&=\sum_{\substack{j(j \neq i)\\ |\bm{r}_{ij}| \leq r_c}} \mathbb{E} \left[ (\tilde{\bm{\sigma}}^{es}_{ij,1}+\tilde{\bm{\sigma}}^{es}_{ij,2})\right]+\frac{N_i^s}p\sum_{\substack{j(j \neq i)\\  r_c < |\bm{r}_{ij}| \leq r_s}} \mathbb{E} \left[ (\tilde{\bm{\sigma}}^{es}_{ij,1}+\tilde{\bm{\sigma}}^{es}_{ij,2}) H'_{ij}\right]-{\bm{\sigma}}_i^{e}  =0,
  \end{split}
\end{equation} 
which completes the proof of $\mathbb{E}[\bm{\mathcal{X}}_i]=0$ for a linear function ${\mathcal{F}'(\rho)}$.

If ${\mathcal{F}'(\rho)}$ is not a linear function, then we can rewrite ${\mathcal{F}'(\tilde{\bar{\rho}}_i)}$ as 
  \begin{equation}\label{taylorexpantionofF}
  {\mathcal{F}'(\tilde{\bar{\rho}}_i)} = \mathcal{F}'(\bar{\rho}_i) + \mathcal{F}''(\bar{\rho}_i) \cdot (\tilde{\bar{\rho}}_i -{\bar{\rho}}_i)+{\cal O}(\tilde{\bar{\rho}}_i -{\bar{\rho}}_i)^2,
  \end{equation}
  which implies that $  \mathbb{E} \bigl[ {\mathcal{F}'(\tilde{\bar{\rho}}_i)}\bigr]
     =  \mathcal{F}'\bigl(\bar{\rho}_i \bigr)+{\cal O}(\mathbb{V}[{\mathcal Q}_i])$. 
     Repeating the proof process for a linear function ${\mathcal{F}'(\tilde{\bar{\rho}}_i)}$, we have $\mathbb{E}[\bm{\mathcal{X}}_i]={\cal O}(\mathbb{V}[{\mathcal Q}_i])$.
     According to Remark~\ref{remark-1}, $\mathbb{V}[{\mathcal Q}_i]$ is close to zero for $r_c>r_{\rm{2NN}}$. Therefore, our approximation for force $\bm{\sigma}_i$ is close to unbiased.

Next, we estimate the variance of $\bm{\mathcal{X}}_i$. At the beginning, we still assume that ${\mathcal{F}'(\tilde{\bar{\rho}}_i)}$ is a linear function. Due to $\mathbb{E}[\bm{\mathcal{X}}_i]=0$, we have
\begin{equation}\label{eq:var_sigma}
  \begin{split}
    \mathbb{V}[\bm{\mathcal{X}}_i] 
    &  
      =  \mathbb{E}\bigl[ \tilde{\bm{\sigma}}^2_i+{\bm{\sigma}}^2_{ave}-2\tilde{\bm{\sigma}}_i{\bm{\sigma}}_{ave}-{\bm{\sigma}}^2_i\bigr]. 
  \end{split}
\end{equation}
Let $\hat{\sigma}_i=\tilde{\bm{\sigma}}_i^{ec} + \tilde{\bm{\sigma}}_i^{es}+ \tilde{\bm{\sigma}}_i^{ps}$. 
According to the force balance condition in \eqref{forcebalance}, $\sigma_{ij}^p=-\sigma_{ji}^p$, and the mutual independence of forces on atoms, we obtain   
\begin{equation}
\label{eq:var_sigma-1}
  \begin{split}
  \mathbb{E}[\tilde{\sigma}_i\sigma_{ave}]&=\frac{1}{N}\mathbb{E}\left[\tilde{\sigma}_i\sum_{k=1}^N(\hat{\sigma}_k+{\sigma}_k^{pc}) \right]=\frac{1}{N}\mathbb{E}\left[\tilde{\sigma}_i\right]\mathbb{E}\left[\sum_{k\neq i}^N(\hat{\sigma}_k+{\sigma}_k^{pc})+ {\sigma}_i^{pc}\right]+\frac{1}{N}\mathbb{E}[\tilde{\sigma}_i\hat{\sigma}_i]\\&=\frac{1}{N}\left\{\mathbb{E}[\hat{\sigma}^2_i]-(\sigma_i-{\sigma}_i^{pc})^2\right\}=\frac{1}{N}\mathbb{V}[\hat{\sigma}_i].
  \end{split}
\end{equation}
Since $ \mathbb{E}[\sigma_{ave}]=0$, the variance of $\sigma_{ave}$ is   
\begin{equation}
\label{eq:var_sigma-1-1}
  \begin{split}
  \mathbb{V}[\sigma_{ave}]=\frac{1}{N^2}\mathbb{E}\left[\sum_{i,j}^N(\hat{\sigma}_i+{\sigma}_i^{pc})(\hat{\sigma}_j+{\sigma}_j^{pc}) \right]=\frac{1}{N^2}\sum_{i}^N(\mathbb{E}[\hat{\sigma}^2_i]-(\sigma_i-{\sigma}_i^{pc})^2)=\frac{1}{N^2}\sum_{i}^N\mathbb{V}[\hat{\sigma}_i].
  \end{split}
\end{equation}
Similarly, the expectation of $\tilde{\sigma}^2_i$ is 
\begin{equation}
\label{eq:var_sigma-1-2}
  \begin{split}
  \mathbb{E}[\tilde{\sigma}^2_i]-{\sigma}^2_i&=\mathbb{E}\left[(\hat{\sigma}_i+{\sigma}_i^{pc})^2 \right]=\mathbb{E}[\hat{\sigma}^2_i]-(\sigma_i-{\sigma}_i^{pc})^2=\mathbb{V}[\hat{\sigma}_i].
  \end{split}
\end{equation}
Combining \eqref{eq:var_sigma-1}, \eqref{eq:var_sigma-1-1}, and \eqref{eq:var_sigma-1-2}, we can obtain
\begin{equation}
\label{eq:var_sigma-1-3}
  \begin{split}
    \mathbb{V}[\bm{\mathcal{X}}_i] 
    &  
      =  \frac{N-2}N\mathbb{V}[\hat{\sigma}_i]+\frac{1}{N^2}\sum_{j=1}^N\mathbb{V}[\hat{\sigma}_j]. 
  \end{split}
\end{equation}

Using the covariance inequality and the Cauchy-Schwarz inequality, the variance of $\hat{\sigma}_i$ is 
\begin{equation}
\label{variance-total-1}
  \begin{split}
  \mathbb{V}[\hat{\sigma}_i]&=\mathbb{V}[\tilde{\sigma}^{ec}_i+\tilde{\sigma}^{es}_i+\tilde{\sigma}^{ps}_i]\leq 3\left(\mathbb{V}[\tilde{\sigma}^{ec}_i]+\mathbb{V}[\tilde{\sigma}^{es}_i]+\mathbb{V}[\tilde{\sigma}^{ps}_i]\right).
  \end{split}
\end{equation}
We now analyze the three variances separately. Following the proof for $ \mathbb{V}[\tilde{\sigma}^{ps}_i]$ in \cite{liang2021RBL}, we have 
\begin{equation}
\label{variance-ps-1}
  \begin{split}
  \mathbb{V}[\tilde{\sigma}^{ps}_i]&\lesssim  \frac{N_i^s-p}{p} \sum_{\substack{j(j\neq i)\\ r_c<|\bm{r}_{ij}| \leq r_s}}
  ({\sigma}^{p}_{ij})^2\leq \frac{(N_i^s-p)N_i^s}{p} \max_{\substack{j(j\neq i)\\ r_c<|\bm{r}_{ij}| \leq r_s}}
  ({\sigma}^{p}_{ij})^2.
  \end{split}
\end{equation}
Due to the form of pairwise function ${\cal V}$, $\mathbb{V}[\tilde{\sigma}^{ps}_i]$ is bounded and decays rapidly as $r_c$ increase.
Similarly,  the variance of $\tilde{\sigma}^{ec}_i$ is 
\begin{equation}\label{variance-ec-1}
  \begin{split}
  \mathbb{V}[\tilde{\sigma}^{ec}_i]&=\mathbb{V}\Bigl[\sum_{\substack{j(j\neq i)\\  |\bm{r}_{ij}| \leq r_c}} 
\left(\tilde{\bm{\sigma}}^{ec}_{ij,1}+ \tilde{\bm{\sigma}}^{ec}_{ij,2}\right)\Bigr]\leq  2 N_i^c\sum_{\substack{j(j\neq i)\\  |\bm{r}_{ij}| \leq r_c}}\left(\mathbb{V} 
\left[\tilde{\bm{\sigma}}^{ec}_{ij,1}\right]+ \mathbb{V}\left[\tilde{\bm{\sigma}}^{ec}_{ij,2}\right]\right)
\\&\leq 2 N_i^c\sum_{\substack{j(j\neq i)\\  |\bm{r}_{ij}| \leq r_c}}\left(\big[\mathcal{F}''(\bar{\rho}_i) {\rho}'(|\bm{r}_{ij}|)\big]^2\mathbb{V} 
\left[{\cal Q}_i\right]+\big[\mathcal{F}''(\bar{\rho}_j) {\rho}'(|\bm{r}_{ji}|)\big]^2\mathbb{V} 
\left[{\cal Q}_j\right]\right)\\&
\lesssim 4 (N_i^c)^2 \max_{\substack{j(j\neq i)\\ |\bm{r}_{ij}| \leq r_c}}
  \mathbb{V} 
\left[{\cal Q}_j\right],
  \end{split}
\end{equation}
which is bounded and decays exponentially with the increasing $r_c$. Here $N_i^c$ is the number of atoms in the core of $i$-th atom.

If the $j$-th atom is in the shell of the $i$-th atom, we have 
\begin{equation}
\begin{aligned}\mathbb{V} \bigl[\tilde{\bm{\sigma}}^{e}_{ij} H'_{ij}\bigr]&=\mathbb{E} \bigl[(\tilde{\bm{\sigma}}^{e}_{ij})^2 H'_{ij}\bigr]
-\mathbb{E} \bigl[\tilde{\bm{\sigma}}^{e}_{ij} H'_{ij}\bigr]^2
=\frac{p}{N_i^s}\mathbb{E} \bigl[(\tilde{\bm{\sigma}}^{e}_{ij})^2 \bigr]
-\frac{p^2}{(N_i^s)^2}\mathbb{E} \bigl[\tilde{\bm{\sigma}}^{e}_{ij} \bigr]^2\\
&=\frac{p}{N_i^s}\mathbb{V}\bigl[\tilde{\bm{\sigma}}^{e}_{ij} \bigr]+\frac{p(N_i^s-p)}{(N_i^s)^2}({\bm{\sigma}}^{e}_{ij})^2.
\end{aligned}
\end{equation}
Similar to \eqref{variance-ec-1}, 
we can prove $\mathbb{V}\bigl[\tilde{\bm{\sigma}}^{e}_{ij} \bigr]\lesssim 2 (\mathbb{V} 
\left[{\cal Q}_i\right]+\mathbb{V} 
\left[{\cal Q}_j\right])$.
Then we get 
\begin{equation}
\label{variance-es-1}
\begin{aligned}\mathbb{V} \bigl[\tilde{\bm{\sigma}}^{es}_{i}\bigr]&=\frac{(N_i^s)^2}{p^2}\mathbb{V}\Bigl[\sum_{\substack{j(j\neq i)\\  r_c<|\bm{r}_{ij}| \leq r_s}} \tilde{\bm{\sigma}}^{e}_{ij}H'_{ij}\Bigr]\leq \frac{(N_i^s)^3}{p^2}\sum_{\substack{j(j\neq i)\\  r_c<|\bm{r}_{ij}| \leq r_s}}\mathbb{V}\Bigl[ \tilde{\bm{\sigma}}^{e}_{ij}H'_{ij}\Bigr]\\
&=\frac{(N_i^s)^2}{p}\sum_{\substack{j(j\neq i)\\  r_c<|\bm{r}_{ij}| \leq r_s}}\mathbb{V}\Bigl[ \tilde{\bm{\sigma}}^{e}_{ij}\Bigr]+\frac{N_i^s(N_i^s-p)}{p}\sum_{\substack{j(j\neq i)\\  r_c<|\bm{r}_{ij}| \leq r_s}} ({\bm{\sigma}}^{e}_{ij})^2\\
&\lesssim   \frac{2(N_i^s)^2}{p}\sum_{\substack{j(j\neq i)\\  r_c<|\bm{r}_{ij}| \leq r_s}}\Bigl(\mathbb{V} 
\left[{\cal Q}_i\right]+\mathbb{V}\left[{\cal Q}_j\right]\Bigr)+\frac{(N_i^s)^2(N_i^s-p)}{p}\max_{\substack{j(j\neq i)\\ r_c<|\bm{r}_{ij}| \leq r_s}}
  (\bm{\sigma}^{e}_{ij})^2\\
&\lesssim \frac{4(N_i^s)^3}{p}\max_{\substack{j(j\neq i)\\ |\bm{r}_{ij}| \leq r_s}}
  \mathbb{V} 
\left[{\cal Q}_j\right] +\frac{(N_i^s)^2(N_i^s-p)}{p}\max_{\substack{j(j\neq i)\\ r_c<|\bm{r}_{ij}| \leq r_s}}
  (\bm{\sigma}^{e}_{ij})^2,
\end{aligned}
\end{equation}
which is also bounded and decays exponentially with the increasing $r_c$.
Finally, substituting (\ref{variance-total-1}-\ref{variance-es-1}) into \eqref{eq:var_sigma-1-3}, we conclude that $\mathbb{V}[\bm{\mathcal{X}}_i]$ is bounded and decays rapidly as $r_c$ increase.

If ${\cal F}'(\cdot)$ is not a linear function, the calculation of $\mathbb{V}[\bm{\mathcal{X}}_i]$ is quite complicated. Let us denote $\bm{\mathcal{X}}_i=\bm{\mathcal{X}}_i^1+\bm{\mathcal{X}}_i^2$, where $\bm{\mathcal{X}}_i^1$ denotes the contribution of $ \mathcal{F}'(\bar{\rho}_i) + \mathcal{F}''(\bar{\rho}_i) \cdot (\tilde{\bar{\rho}}_i -{\bar{\rho}}_i)$ to $\bm{\mathcal{X}}_i$. 
Repeating the proof for a linear linear function ${\cal F}'(\cdot)$, we find that $\mathbb{V}[\bm{\mathcal{X}}_i^1]$  is bounded and decays rapidly with the increase in $r_c$. 
From \eqref{taylorexpantionofF} and the relationship between ${\cal F}$ and $\bm{\sigma}^{e}$, we observe that $\bm{\mathcal{X}}_i^1$ is a random quantity with exponentially decaying variance. 
Therefore, we conclude that $\mathbb{V}[\bm{\mathcal{X}}_i]$ is bounded and decays rapidly with the increasing $r_c$. 
This completes the proof of the theorem.
\end{proof}

\begin{figure}[htbp]
  \centering
  \subfloat[FCC]{\includegraphics[width = 0.3\textwidth]{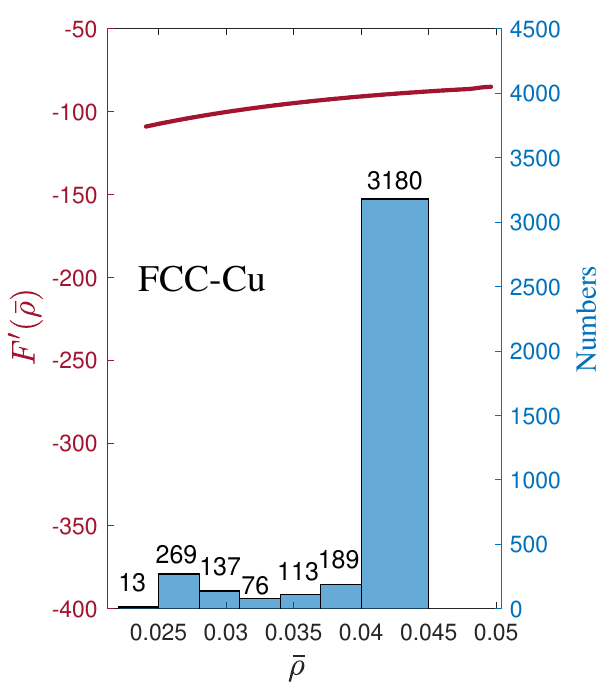}}
  \subfloat[BCC]{\includegraphics[width = 0.3\textwidth]{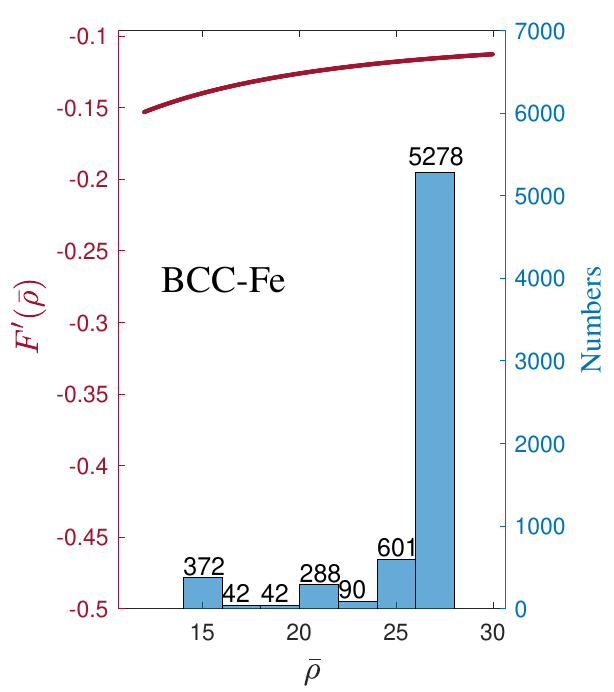}}
  \subfloat[HCP]{\includegraphics[width = 0.3\textwidth]{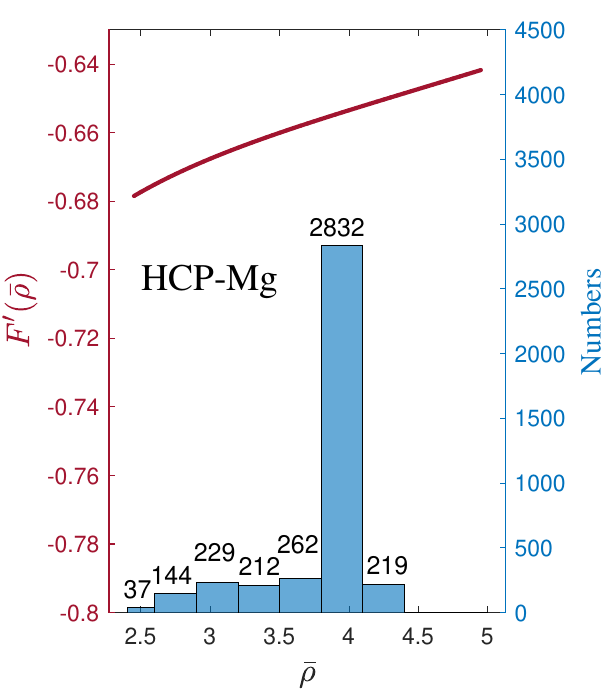}}
  \caption{The solid lines represent the distributions of the function $\mathcal{F}'(\bar{\rho})$, while the histogram depicts the density distribution of all atoms within the given systems.}\label{fig:F_drho_FCC-BCC-HCP}
\end{figure}
\begin{remark}\label{remark-2}
In the EAM potential, the function ${\mathcal{F}'(\rho)}$ is typically not strictly linear for $\rho > 0$. 
However, in most applications, when focusing on the range of actual density values, it can be approximated as nearly linear, as illustrated in \autoref{fig:F_drho_FCC-BCC-HCP}. This observation supports the plausibility of our proposed linear function assumption for ${\mathcal{F}'(\rho)}$.
\end{remark}

\subsection{RBL-EAM method for Newton-Pair system}
In MD simulations, the computational cost of force calculations can be reduced by half when the interaction forces adhere to Newton's third law. Such a system is typically referred as a Newton-pair system. 
For metallic systems described by the EAM potential, as demonstrated by \eqref{eq:sigma_i_components-p} and \eqref{forcebalance}, we observe that these systems exhibit Newton-pair behavior. 
Consequently, the computational cost of force calculations can be further minimized by integrating the Newton-pair property into the RBL-EAM method.
To achieve this, we first follow the rule used in LAMMPS to construct a half neighbor list. As illustrated in \autoref{fig:neigh_list}, let us consider a system with 14 atoms. The full neighbor list for the 6-th atom includes all other atoms in the system except itself. 
Furthermore, the 6-th atom also appear in the neighbor list of other atoms. 
However, the half neighbor list for the 6-th atom includes only atoms with indices $j>6$, but the 6-th atom would not appear in the neighbor lists of atoms with indices $j>6$. This reduces redundancy in force calculations.
\begin{figure}[htbp]
  \centering
  \subfloat[Full neighbor list]{\includegraphics[width = 0.45\textwidth]{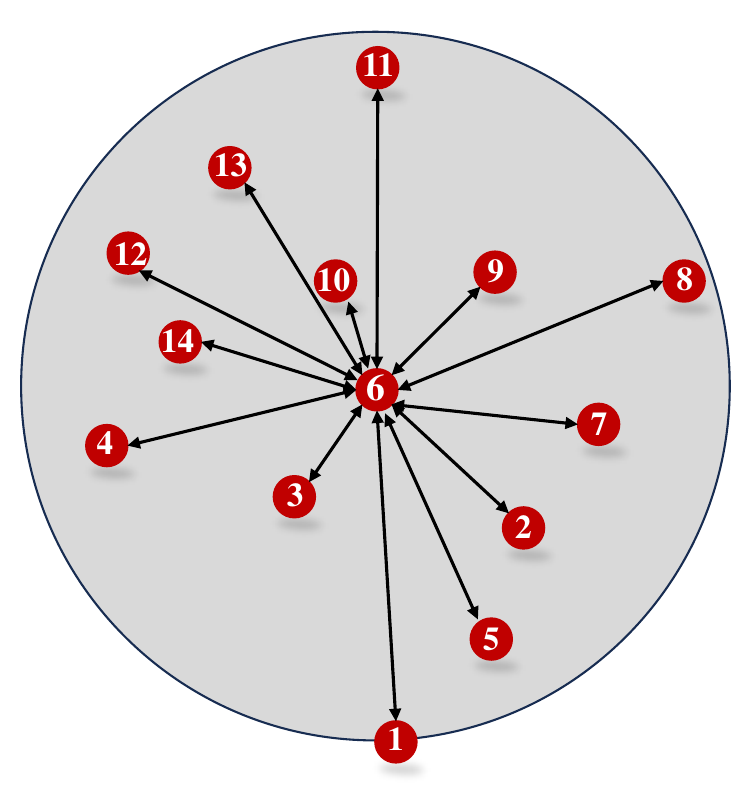}} \quad
  \subfloat[Half neighbor list]{\includegraphics[width = 0.45\textwidth]{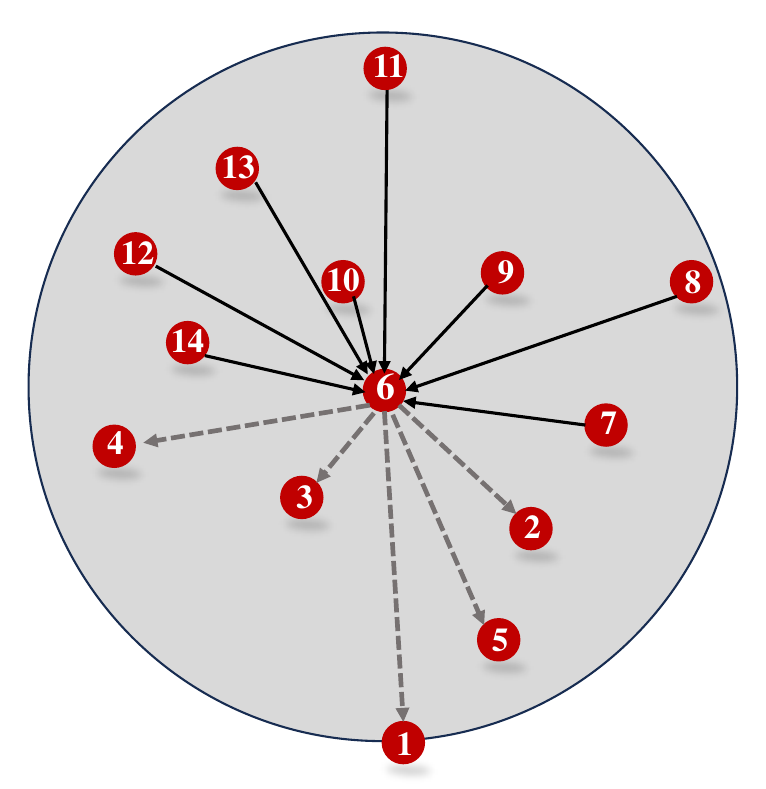}}
  \caption{Sketch of the (a) full and (b) half neighbor list for the 6th atom. Here an arrow pointing to atom $i$ means that atom $j$ belongs to the neighbor list of atom $i$.}\label{fig:neigh_list}
\end{figure}

For a Newton-pair system, we further categorize the atoms in the half neighbor list into two types: (1) I-type atoms within the small core region, defined as $\mathfrak{R}_i^{c} = \bigl\{j\big| |\bm{r}_{ij}|\leq r_c \bigr\}$; and (2) II-type atoms outside of the core but within the shell region,  defined as $\mathfrak{R}_i^{s} = \bigl\{j\big| r_c<|\bm{r}_{ij}|\leq r_s +\Delta r \bigr\}$.
In the case of half neighbor list, if $|\bm{r}_{ij}|\leq r_c$, then exactly one of the following conditions holds:  $j\in \mathfrak{R}_i^{c}$ and $i\in \mathfrak{R}_j^{c}$. 
Similarly, if $ r_c<|\bm{r}_{ij}|\leq r_s+\Delta r$, then exactly one of the following conditions holds:  $j\in \mathfrak{R}_i^{s}$ and $i\in \mathfrak{R}_j^{s}$. 
Using the half neighbor list, the host electron density is calculated as 
\begin{align} \label{eq:rho_c-np}
  \bar{\rho}_i = \sum_{j\in\mathfrak{R}_i^{c}\cup\mathfrak{R}_i^{s}} \rho(|\bm{r}_{ij}|)+\sum_{\substack{j'(j'\neq i)\\ i\in\mathfrak{R}_{j'}^{c}\cup\mathfrak{R}_{j'}^{s}}} \rho(|\bm{r}_{j'i}|).
\end{align}
The second term on the right-hand side of \eqref{eq:rho_c-np} can be computed by copying the results from the $j'$-th atom, thereby reducing the computational cost. Similarly, the interaction force ${\bm{\sigma}}_i$ is obtained as following
 \begin{align}\label{eq:sigma_i-np}
  \begin{split}
    \bm{\sigma}_i  = \sum_{j\in\mathfrak{R}_i^{c}\cup\mathfrak{R}_i^{s}}\left(\bm{\sigma}^{e}_{ij,1}+ \bm{\sigma}^{e}_{ij,2}+\bm{\sigma}^{p}_{ij}\right)-\sum_{\substack{j'(j'\neq i)\\ i\in\mathfrak{R}_{j'}^{c}\cup\mathfrak{R}_{j'}^{s}}}\left(\bm{\sigma}^{e}_{j'i,2}+ \bm{\sigma}^{e}_{j'i,1}+\bm{\sigma}^{p}_{j'i}\right).
  \end{split} 
\end{align}
The second term on the right-hand side of \eqref{eq:sigma_i-np} can be computed by copying the results from the $j'$-th atom and then change the sign.

Next, we apply the RBL-EAM method to the Newton-pair system. For the $i$-th atom, we randomly select $p$ atoms from $\mathfrak{R}_i^s$ to form the batch $\mathfrak{B}_i$. The host electron density is approximated as 
\begin{align} \label{eq:rho_c_n}
  \tilde{\bar{\rho}}_i := {\bar{\rho}}_i^c+\tilde{\bar{\rho}}_i^s=\sum_{j\in\mathfrak{R}_i^{c}} \rho(|\bm{r}_{ij}|)+\sum_{\substack{j'(j'\neq i)\\i\in\mathfrak{R}_{j'}^{c}}} \rho(|\bm{r}_{j'i}|)+\frac{N_i^s}{p} \sum_{j \in \mathfrak{B}_i} \rho(|\bm{r}_{ji}|) +  \sum_{\substack{j'(j'\neq i)\\i \in \mathfrak{B}_{j'}}} \frac{N_{j'}^s}{p}\rho(|\bm{r}_{j'i}|), 
\end{align}
where $N_i^s=|\mathfrak{R}_i^{s}|$ and $N_{j'}^s=|\mathfrak{R}_{j'}^{s}|$.
Let $\mathcal{Q}_i = \tilde{\bar{\rho}}_i - \bar{\rho}_i$, we can further prove in \autoref{theorem-1-n} that this approximation for the host electron density is unbiased, namely, the expectation of $\mathcal{Q}_i$ is zero and its variance is bounded.

%
\begin{theorem}\label{theorem-1-n}
 The expectation of  $\mathcal{Q}_i$ is zero and its variance is bounded.
\end{theorem}
\begin{proof}
By introducing the indicative functions $H_{ij}$ and $H_{j'i}$ with $\mathbb{E}[H_{ij}]=p/N_i^s$ and $\mathbb{E}[H_{j'i}]=p/N_{j'}^s$, the expectation of $\tilde{\bar{\rho}}_i^s$ is calculated as 
\begin{equation}\label{eq:Expec-rho-n}
  \mathbb{E} [\tilde{\bar{\rho}}_i^{s}] = \frac{N_i^s}{p} \sum_{j\in\mathfrak{R}^s_i} \rho(|\bm{r}_{ij}|) \mathbb{E} [H_{ij}]
  +  \sum_{\substack{j'(j'\neq i)\\i\in\mathfrak{R}^s_{j'}}} \frac{N_{j'}^s}{p}\rho(|\bm{r}_{j'i}|) \mathbb{E} [H_{j'i}]
  = \bar{\rho}_i^s.
\end{equation}
According to the derivation in \eqref{eq:Expec-Q}, it is clear that $\mathbb{E}[\mathcal{Q}_i] = 0$. Similar to \eqref{eq:17}, we can further calculate the variance of the quantity $\mathcal{Q}_i$ as $\mathbb{V}[{\mathcal Q}_i] = \mathbb{V}[\tilde{\bar{\rho}}^{s}_i - {\bar{\rho}}^s_i]  = \mathbb{E}\bigl[ (\tilde{\bar{\rho}}_i^{s})^2\bigr] - ({\bar{\rho}}_i^s)^2$, where 
\begin{align}
  \begin{split}
    \mathbb{E}\bigl[ (\tilde{\bar{\rho}}_i^{s})^2\bigr] 
    & = \frac{(N_i^s)^2}{p^2} \mathbb{E}\biggl[
     \Bigl( \sum_{j \in \mathfrak{R}^s_i} \rho(|\bm{r}_{ij}|) H_{ij}\Bigr)^2
    \biggr] + \mathbb{E}\biggl[
     \Bigl( \sum_{\substack{j'(j'\neq i)\\i\in\mathfrak{R}^s_{j'}}}\frac{N_{j'}^s}{p} \rho(|\bm{r}_{j'i}|) H_{j'i}\Bigr)^2
    \biggr] \\
    &\quad + 2\frac{N_i^s }{p^2} \mathbb{E}\biggl[
     \Bigl( \sum_{j \in \mathfrak{R}^s_i} \rho(|\bm{r}_{ij}|) H_{ij}\Bigr)  
     \Bigl( \sum_{\substack{j'(j'\neq i)\\i\in\mathfrak{R}^s_{j'}}} N_{j'}^s \rho(|\bm{r}_{j'i}|) H_{j'i}\Bigr)
    \biggr]\\
    &\leq \frac{N_i^s}{p} \Bigl( \sum_{j \in \mathfrak{R}^s_i} \rho(|\bm{r}_{ij}|) \Bigr)^2 + \frac{\max_{j'}N_{j'}^s}{p}\Bigl( \sum_{\substack{j'(j'\neq i)\\i\in\mathfrak{R}^s_{j'}}}\rho(|\bm{r}_{j'i}|) \Bigr)^2
    \\&\quad+2
     \Bigl( \sum_{j \in \mathfrak{R}^s_i} \rho(|\bm{r}_{ij}|) \Bigr)  
     \Bigl( \sum_{\substack{j'(j'\neq i)\\i\in\mathfrak{R}^s_{j'}}} N_{j'}^s \rho(|\bm{r}_{j'i}|) \Bigr)\leq \frac{\max_{j} N_j^s}{p} ({\bar{\rho}}_i^s)^2. 
  \end{split}
  \end{align}
Here we use the fact $\mathbb{E}[H_{ij}H_{j'i}]=\frac{p^2}{N_i^sN_{j'}^s}$ and $\mathbb{E}[H_{j''i}H_{j'i}]=\frac{p^2}{N_{j''}^sN_{j'}^s}$, which hold due to the independence of the random processes. 
Following the proof of \autoref{theorem-1}, we complete the proof of this theorem.
\end{proof}

Similar to the host electron density $\tilde{\bar{\rho}}_i$, the approximate interaction force is given by
\begin{equation}\label{eq:total_sigma_n}
  \tilde{\bm{\sigma}}_i =\tilde{\bm{\sigma}}_i^{ec}+\tilde{\bm{\sigma}}_i^{es}+{\bm{\sigma}}_i^{pc} + \tilde{\bm{\sigma}}_i^{ps},
\end{equation}
where each component is expressed as follows:
\begin{subequations}
  \begin{align}
   \tilde{\bm{\sigma}}_i^{ec} &= 
\sum_{j\in\mathfrak{R}_i^{c}}\left(\tilde{\bm{\sigma}}^{e}_{ij,1}+ \tilde{\bm{\sigma}}^{e}_{ij,2}\right)-\sum_{\substack{j'(j'\neq i)\\ i\in\mathfrak{R}_{j'}^{c}}}\left(\tilde{\bm{\sigma}}^{e}_{j'i,2}+ \tilde{\bm{\sigma}}^{e}_{j'i,1}\right)   ,\label{eq:sigma_ec-n}\\ 
   \tilde{\bm{\sigma}}_i^{es} & 
   =  \frac{N_i^s}{p}\sum_{j\in\mathfrak{B}'_i}\left( \tilde{\bm{\sigma}}^{e}_{ij,1}+  \tilde{\bm{\sigma}}^{e}_{ij,2}\right) 
   - \sum_{\substack{j'(j'\neq i)\\ i\in\mathfrak{B}'_{j'}}} \left(\frac{N_{j'}^s}{p}\left(\tilde{\bm{\sigma}}^{e}_{j'i,1}+ \tilde{\bm{\sigma}}^{e}_{j'i,2}\right)\right), \label{eq:sigma_es-n}\\ 
    \bm{\sigma}_i^{pc} & = \sum_{j\in\mathfrak{R}_i^{c}}{\bm{\sigma}}^{p}_{ij}-\sum_{\substack{j'(j'\neq i)\\ i\in\mathfrak{R}_{j'}^{c}}}{\bm{\sigma}}^{p}_{j'i}, \label{eq:sigma_pc-n}
    \\
    \tilde{\bm{\sigma}}_i^{ps} & = \frac{N_i^s}{p} \sum_{j \in \mathfrak{B}'_i}  \bm{\sigma}^{p}_{ij} - \sum_{\substack{j'(j'\neq i)\\ i\in\mathfrak{B}'_{j'}}} \frac{N_{j'}^s}{p} \bm{\sigma}^{p}_{j'i}. \label{eq:sigma_ps-n}
  \end{align}
\end{subequations}
Here $\tilde{\bm{\sigma}}^{e}_{ij,1}=-\mathcal{F}'(\tilde{\bar{\rho}}_i) {\rho}'(|\bm{r}_{ij}|)\frac{\partial |\bm{r}_{ij}|}{\partial \bm{q}_i}$ and $\tilde{\bm{\sigma}}^{e}_{ij,2}=- \mathcal{F}'(\tilde{\bar{\rho}}_j) {\rho}'(|\bm{r}_{ji}|)\frac{\partial |\bm{r}_{ji}|}{\partial \bm{q}_i}$. 
For a given batch size $p$, we randomly select $p$ atoms from $\mathfrak{R}_i^{s}$, resulting in $\mathfrak{B}_i'$. 
It should be noted that this random processor is independent of the one involved in the calculation of $\tilde{\bar{\rho}}_i$. 
Due to the half-list structure, we observe that the net force $\bm{\sigma}_{ave}$ is always equals to zero.
Let $\bm{\mathcal{X}}_i = \tilde{\bm\sigma}_i - \bm{\sigma}_i$, where $\bm{\sigma}_i$ is the interacting force defined in \eqref{eq:sigma_i} or \eqref{eq:sigma_i-np} with cutoff radius $r_s$. 
We establish the following theorem by following the ideas of  \autoref{theorem-2} and \autoref{theorem-1-n}. 
The proof of this theorem is omitted here for brevity.
\begin{theorem}\label{theorem-2-n}
 The expectation of $\bm{\mathcal{X}}_i$ is zero for linear function ${\cal F}'(\cdot)$ and close to zero otherwise.  Moreover, the variance of $\bm{\mathcal{X}}_i$  is bounded.
\end{theorem}

\begin{remark}
In \eqref{eq:sigma_es} and \eqref{eq:sigma_es-n}, a new random batch set $\mathfrak{B}_i'$ is introduced to compute $\tilde{\bm{\sigma}}_i^{es}$, which incurs additional computational cost. 
A natural approach to avoid this cost is to set $\mathfrak{B}_i'=\mathfrak{B}_i$. 
However, as shown in \ref{A-1}, this choice prevents us from rigorously proving that the force calculated by the RBL-EMA method is unbiased, even when $ {\cal F}'(\cdot)$ is a linear function. 
Nonetheless, numerical simulations in \ref{sec:4} indicate that the accuracy of RBL-EMA method remains acceptable even when $\mathfrak{B}_i'=\mathfrak{B}_i$. 
Therefore, to enhance the efficiency of the proposed approach, we will consistently set $\mathfrak{B}_i'=\mathfrak{B}_i$ throughout the remainder of this paper.
\end{remark}

Let us define the Wasserstein-2 distance \cite{santambrogio2015Wasserstein} as:
\begin{equation}
{\mathcal W}_2(\varphi,\phi)
=\left(\inf_{\pi\in\Pi(\varphi,\phi)}\int_{\mathbb{R}^3\times\mathbb{R}^3}\|\bm{x}-\bm{y}\|^2\textnormal{d}\pi(\bm{x},\bm{y})\right)^{1/2},
\end{equation}
where $\Pi(\varphi,\phi)$ denotes the set of all joint distributions with marginal distributions $\varphi$ and $\phi$. 
Assume that $\tau$ represents the time step of the velocity-Verlet algorithm for the underdamped Langevin dynamics. 
The following theorem demonstrates that our method is effective in capturing finite-time dynamics.
\begin{theorem}\label{dynamicproof}
  Let $(\bm{Q}_i, \bm{P}_i)$ be the solution to the following dynamic system:
  \begin{subequations}
  \label{SDE1}
    \begin{align}
      m_i \mathrm{d}{\bm{Q}}_i & = \bm{P}_i \mathrm{d} t, \\
      \mathrm{d} {\bm{P}}_i & = (\bm{\sigma}_i - \gamma \bm{P}_i)\mathrm{d} t+\xi \mathrm{d} W_i, 
      \end{align}
  \end{subequations}
  where $\{W_i\}_{i=1}^N$ are i.i.d. Wiener processes, $\bm{Q}_i$ and $\bm{P}_i$ denote the coordinate and momentum of the $i$-th atom, respectively. The processes generated by the RBL-EAM method satisfy the following stochastic differential equations (SDEs):
  \begin{subequations}\label{SDE2}
    \begin{align}
      m_i \mathrm{d}{\tilde{\bm{Q}}}_i & = \tilde{\bm{P}}_i \mathrm{d} t, \\
      \mathrm{d}{\tilde{\bm{P}}}_i & = (\tilde{\bm{\sigma}}_i - \gamma \tilde{\bm{P}}_i)\mathrm{d} t+\xi \mathrm{d} W_i. 
      \end{align}
  \end{subequations}
  Suppose that both the above equations share the same initial data, and let $\bm{R}$ denote the initial configuration of the system. 
  Assume that the mass of each atom is bounded, the force $\bm{\sigma}_i$ is bounded and Lipschitz continuous, and the variances $\mathcal{Q}_i$ are bounded by a constant $C_1$.
  For any time $t^* > 0$, there exists a constant $C(t^*)$ independent of $N$ such that the Wasserstein distance satisfies
  \begin{equation}
   \sup_{\bm{R}} {\mathcal W}_2(Q(\bm{R},\cdot),\tilde{Q}(\bm{R},\cdot))\leq C(t^*)\sqrt{C_2 \tau +(1+\xi/2 + C_1 )\tau^2},
  \end{equation}
  where $C_2= \|\mathbb{E}[\bm{\mathcal{X}}^2_i]\|_\infty$, $Q(\bm{R},\cdot)$ and $\tilde{Q}(\bm{R},\cdot)$ represent the transition probabilities of the SDEs of the direct truncation method \eqref{SDE1} and the RBL-EAM method \eqref{SDE2}, respectively.
\end{theorem}

The proof of \autoref{dynamicproof} follows a similar approach to that of Theorem II.1 in \cite{liang2021RBL}. Due to space constraints, we omit the details here.
\section{Numerical Examples}\label{sec:4}
In this section,  several benchmark tests, including the calculation of the lattice constant, the radial distribution function, and the elastic constants for metallic system, are implemented to verify the accuracy and efficiency of the newly proposed RBL-EAM method. 
This RBL-EAM method is implemented on top of the LAMMPS framework. All simulations are performed on a Linux system equipped with 32-core Intel(R) Xeon(R) Platinum 8358 CPUs and 32 GB of memory.
For each benchmark test, we consider three metals: Cu, $\alpha$-Fe, and Mg, which correspond to the FCC, BCC, and HCP lattices, respectively. The EAM potentials employed for these metals are “Cu$\_$u6.eam” \cite{Foiles_Cu}, “Fe$\_$mm.eam.fs” \cite{MendelevbccFe}, and “Mg$\_$mm.eam.fs” \cite{WilsonhcpMg}, respectively. 
It is worthy noted that the respective cutoff radii (i.e., $r_s$ in this study) specified in these potential files are $4.95$, $5.3$, and $7.5$\AA.
The time step for all simulation is set as $\Delta t = 0.001$ps.

\subsection{Lattice constant}\label{sec:4.1}
The lattice constant (or lattice parameter) refers to the physical dimensions and angles that define the unit cell geometry in a crystal lattice. This parameter is critical for describing atomic arrangements in crystals and is primarily determined by the balance of interatomic forces. 
As a measurable physical property, the accuracy of the lattice constant serves as an essential benchmark for evaluating how well a potential function models the real atomic interactions. Therefore, we use it to validate the accuracy of the RBL-EAM method.

By treating each atom as a hard sphere with a radius of $r/2$, the equilibrium configuration of a crystal is obtained by minimizing the potential energy:
\begin{equation*}
    r_0 = \arg \min_{r} ~ \bar{\mathcal{U}}(r),
\end{equation*}
where $\bar{\mathcal{U}}(r)$ denotes the average potential energy per atom in the given configuration. 
For FCC and BCC crystals (see the unit cell in \autoref{fig:RVE_Structure}), the lattice constants are given by $a_0:=\sqrt{2} r_0$ and $a_0 := \sqrt{3} r_0/2$, respectively. 
For HCP crystals, the lattice constants are defined as $(a_0,c_0)=(r_0,\sqrt{8/3}r_0)$. 
To determine $r_0$, we compute $\bar{\mathcal{U}}(r_j)$ over a range $r_j\in[\underline{r},\,\overline{r}]$ using the RBL-EAM method. The computation process is summarized in \autoref{Algorithm_3}. 
All simulations are conducted at a temperature of $T=300$K. 
In the RBL-EAM method, we consider batch sizes $p = 0$, $5$, $10$, $15$, $20$, $25$, $30$, and set the core cutoff radius as 
$r_c = 2.8$, $3.1$, $4.0$\AA~for FCC, BCC, and HCP crystals, respectively. The corresponding ratios  $r_c /r_s$ for the three crystals are $0.565$, $0.584$, and $0.533$, respectively. 
\begin{algorithm}[htbp] 
\footnotesize{
  \caption{The computation process of lattice constant. }\label{Algorithm_3}
  \SetAlgoLined
  \SetKwInOut{Initialize}{Initialize}
  \SetKwInOut{Compute}{Compute}
  \KwIn{An EAM potential file for given crystal.}
  \Initialize{Set the core cutoff radius $r_c$ and batch size $p$ for the RBL-EAM method; \\ Set the temperature $T=300$K for all this system; }
  \For{$r_j = \underline{r}\,: \Delta r\,:\overline{r}$}{
       For given $r_j$, build the initial configuration\;
       Relax the system in the NVT ensemble for 10,000 steps to minimize the total potential energy\;
       Calculate the average potential energy $\bar{\mathcal{U}}(r_j)$ for the current configuration.
    }
  \Compute{$r_0 = \arg \min\limits_{r_j} ~ \bar{\mathcal{U}}(r_j) $.}
   \KwOut{The average potential energy $\bar{\mathcal{U}}(r_j)$ and $r_0$.} 
}
\end{algorithm}

The computed average potential energies, $\bar{\mathcal{U}}(r_j)$, for these three crystal structures are presented in \autoref{fig:potential_a0_total}, where they are plotted as functions of $r_j$. 
\begin{figure}[htbp]
  \centering
  \subfloat[Cu-FCC]{\includegraphics[width = 0.85\textwidth]{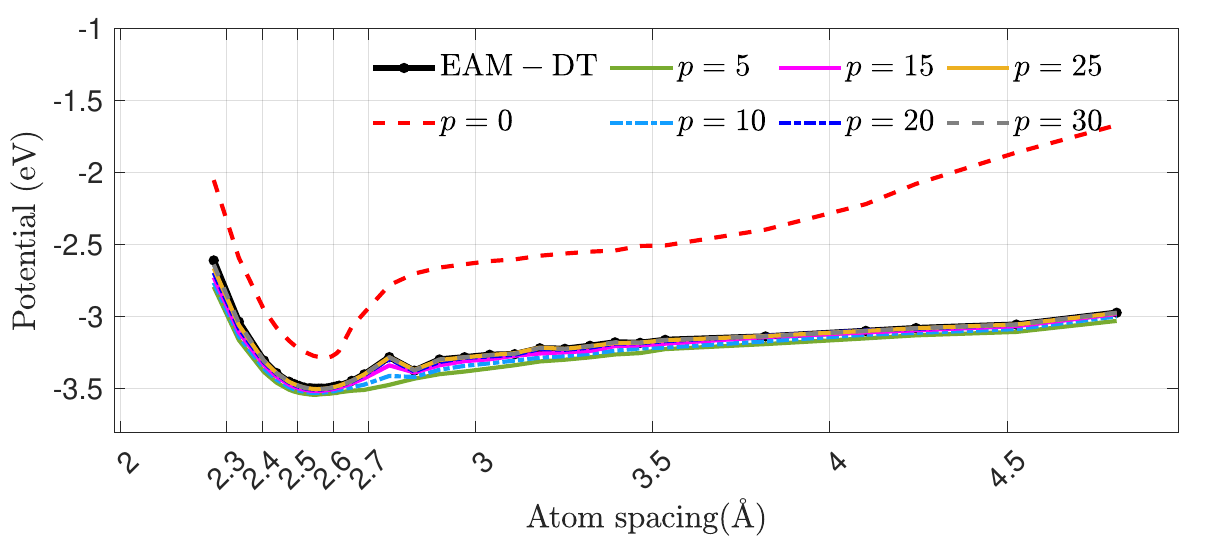}}\\
  \subfloat[$\alpha$-Fe-BCC]{\includegraphics[width = 0.89\textwidth]{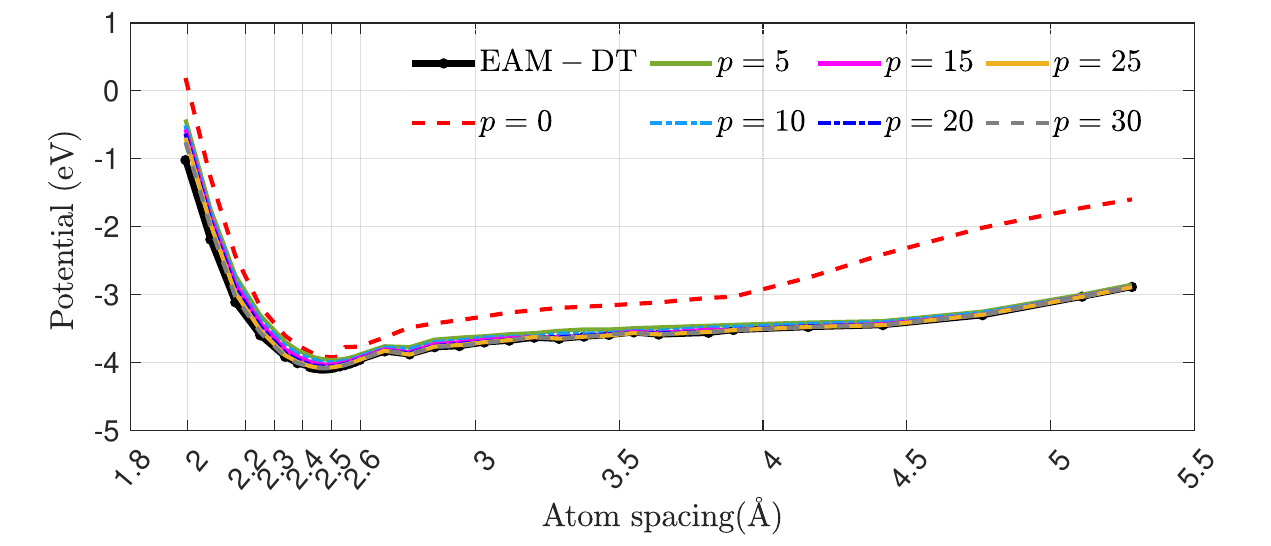}}\\
  \subfloat[Mg-HCP]{\includegraphics[width = 0.85\textwidth]{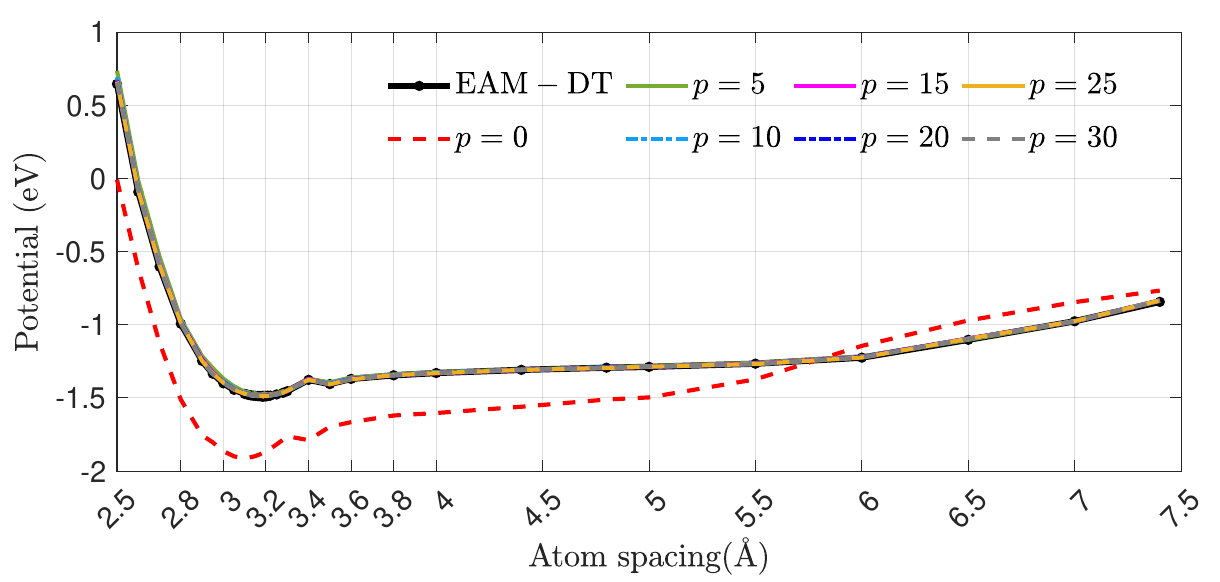}}
  \caption{The curves of the average potential energy $\bar{\mathcal{U}}(r_j)$ for (a) Cu of FCC lattice with $r_c = 2.8$\AA, (b) $\alpha$-Fe of BCC lattice with $r_c = 3.1$\AA, and (c) Mg of HCP lattice with $r_c= 4.0$\AA.  }\label{fig:potential_a0_total}
\end{figure}
The reference solutions are obtained using the EAM potential with direct truncation (EAM-DT). As shown in \autoref{fig:potential_a0_total}, the potential energy curves generated by the RBL-EAM method with $p > 0$ closely match those of the EAM-DT for all three crystals. 
These results indicate that the RBL-EAM method can achieve accurate average potential energies even with a relatively small batch size, such as $p = 5$.
Using the computed average potential energies, we extract the lattice constants for all three crystals, as summarized in \autoref{tab:lattice_a0}. 
Across all batch sizes, the lattice constants obtained with the RBL-EAM method exhibit excellent agreement with the experimental data, with accuracy improving as $p$ increases.

\begin{table}[htbp]
  \centering
  \caption{The lattice constants $a_0$ calculated by the RBL-EAM method with different batch sizes.}\label{tab:lattice_a0}
  \begin{tabular}{ccccccccccc}
  \toprule
    & $p$ & 0 & 5 & 10 & 15 & 20 & 25 & 30 & EAM-DT & Exp.\\
  \midrule   
   Cu& $a_0$  & 3.640  & 3.610   & 3.600  & 3.611  & 3.611   & 3.612  & 3.613  & 3.615     &   \textbf{3.615}      \\ 
    $\alpha$-Fe&$a_0$  & 2.920  & 2.890   & 2.883  & 2.875   & 2.87   & 2.857  &  2.855   & 2.855     &  \textbf{2.850}\\ 
    Mg&$a_0$  & 3.105  & 3.194   & 3.187  & 3.182   & 3.185   & 3.182   &  3.181   &  3.180   &  \textbf{3.184}\\   
  \bottomrule
  \end{tabular}
\end{table}

\subsection{Radial distribution function}\label{sec:4.2}
The radial distribution function (RDF), $g(r)$, is a key statistical measure that characterizes the spatial distribution of atoms within a system. It quantifies how atomic density varies as a function of distance 
${r}$ from a reference atom, providing critical insights into both interatomic potentials and atomic arrangements. 
For example, peaks in the RDF correspond to characteristic atomic distances, revealing underlying crystal lattice structures, 
Then features of the RDF curves, including the height and shape of these peaks, are influenced by interatomic interactions.
Thus, it indicate that the RDF is a valuable tool for assessing the performance of the RBL-EAM method by comparing its RDF curves with reference results.

In MD simulations, the RDF $g(r)$ of a reference atom is calculated as:
\begin{equation}
  g(r) = \frac{1}{\varrho}\frac{\langle n(r)\rangle}{dr},
\end{equation}
where $\langle n(r)\rangle=\frac{n(r)}{4\pi r^2}$ represents the average number of atoms found within the radial interval $[r,r+dr]$.
To evaluate the RDF, the initial MD system configuration is constructed using the experimental lattice constant.
The system is then relaxed and equilibrated in the NVT ensemble for 5,000 steps using the RBL-EAM method. The RDF $g(r)$ is computed based on the relaxed configuration. 
The batch sizes considered for the RBL-EAM method are $p = 0$, $5$, $10$, $20$, and $30$. The core cutoff radius are set as $r_c = 2.8$, $2.7$, and $3.5$\AA~for FCC, BCC, and HCP lattices, respectively. 
The ratios $r_c/r_s$ for these crystals are$ 0.565$, $0.509$, and $0.467$, respectively. 
\begin{figure}[htbp]
  \centering
  \subfloat[Cu-FCC]{\includegraphics[width = 0.8\textwidth]{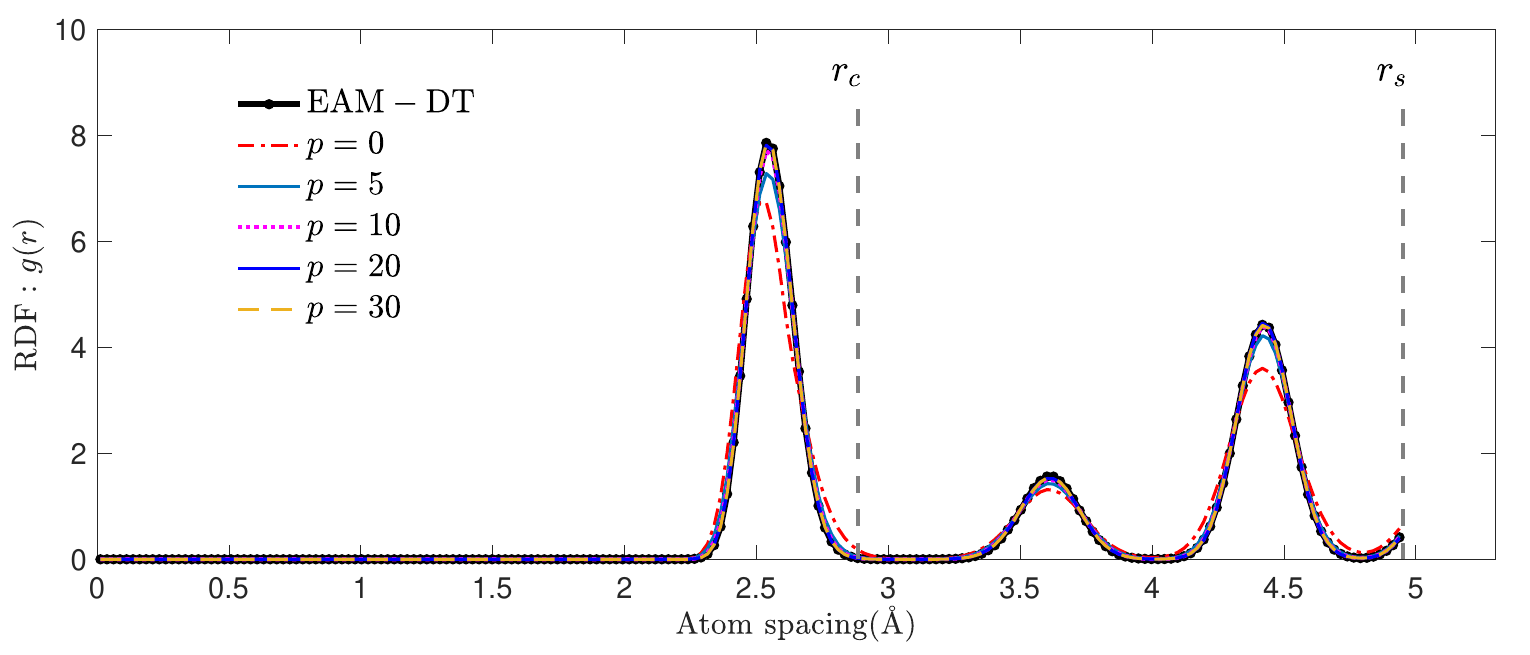}}\\
  \subfloat[$\alpha$-Fe-BCC]{\includegraphics[width = 0.83\textwidth]{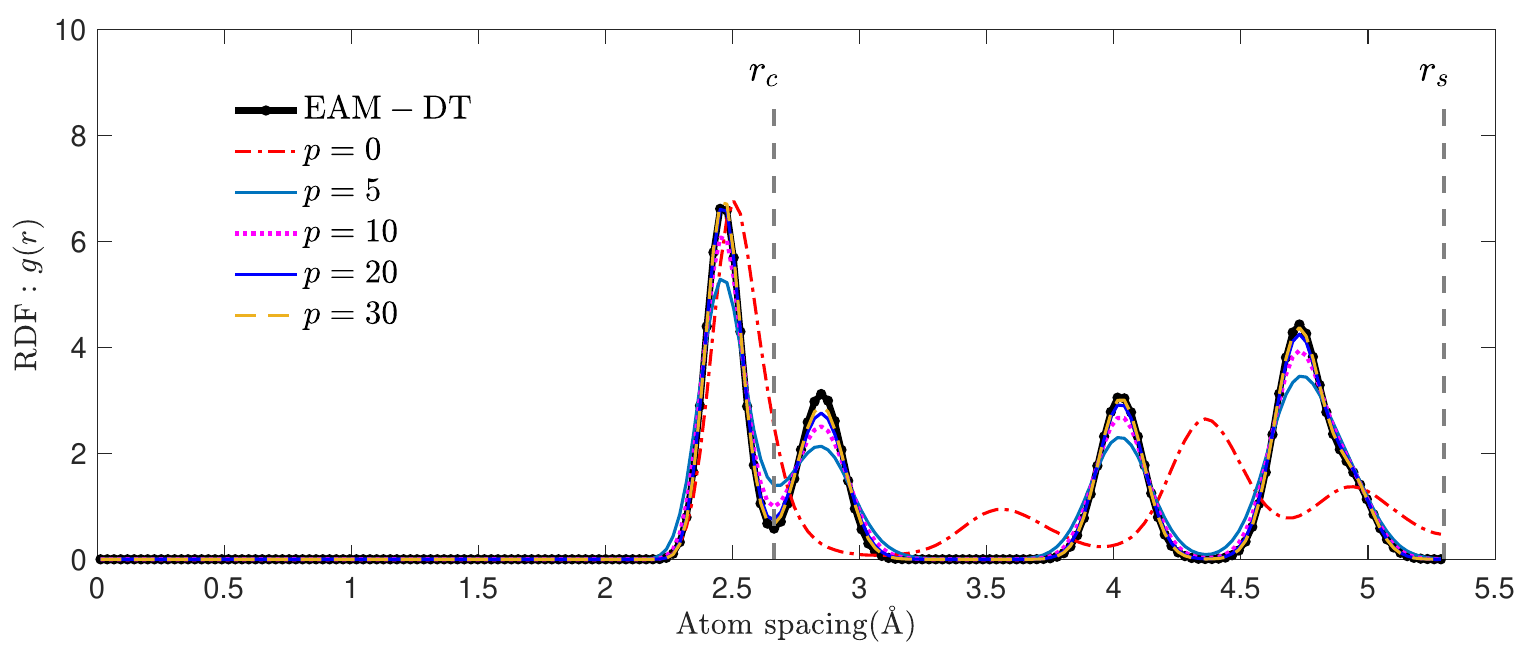}}\\
  \subfloat[Mg-HCP]{\includegraphics[width = 0.8\textwidth]{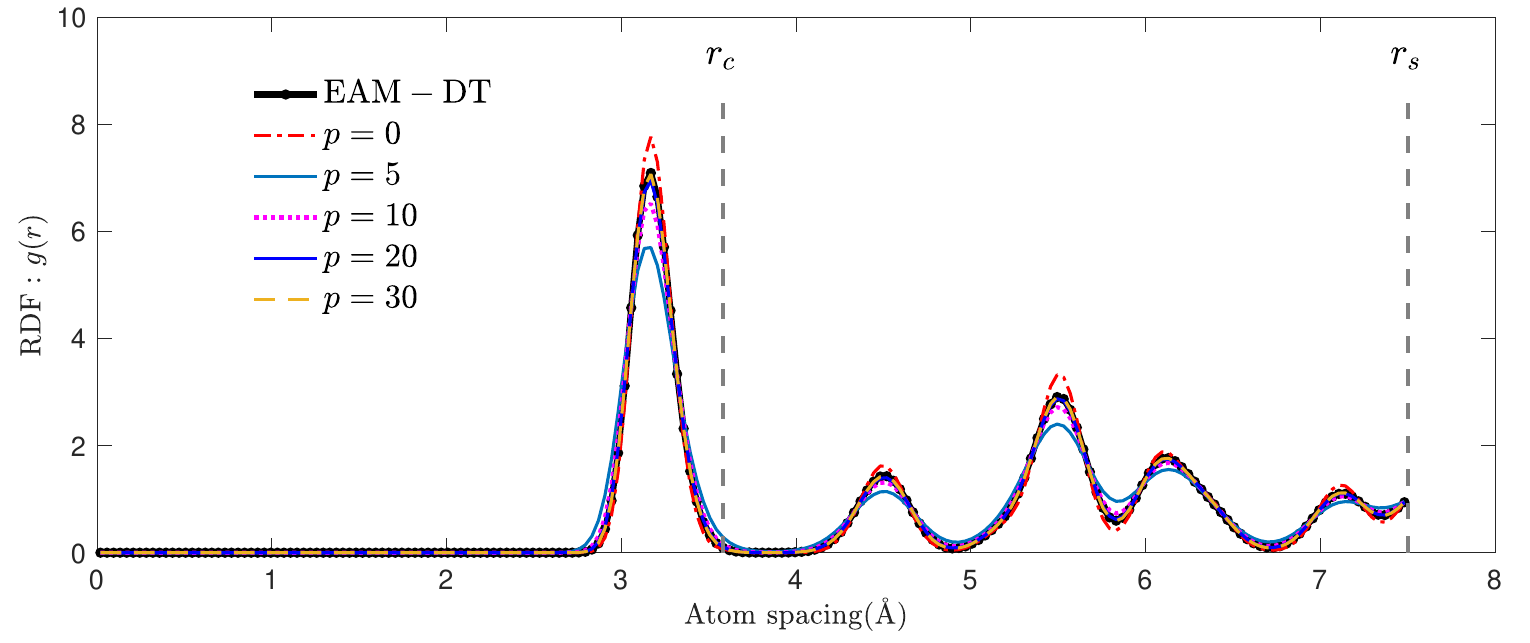}}
  \caption{The RDF curves considering different batch sizes $p = 0, 5, 10, 20, 30$ for (a) Cu of FCC lattice with $r_c = 2.6$\AA; (b) $\alpha$-Fe of BCC lattice with $r_c = 2.7$\AA; and (c) Mg of HCP lattice with $r_c = 3.5$\AA.}\label{fig:RDF-FCC-BCC-HCP}
\end{figure}

The RDF results obtained using the RBL-EAM and EAM-DT methods are presented in \autoref{fig:RDF-FCC-BCC-HCP}.
For the case of $p=0$, noticeable deviations in both peak positions and magnitudes for the RDF curves are observed, particularly for $\alpha$-Fe, when compared to the reference EAM-DT curve.
However, for $p>0$, the peak positions across all lattice types align closely with the reference EAM-DT curve, and as $p$ increases, the peak magnitudes also converge to the EAM-DT values.
Even when the core cutoff radius $r_c$ is near the first nearest-neighboring distance $r_{\rm{1NN}}$, the RBL-EAM method can achieve accurate RDF curves with a small batch size of $p=5$. 
Notably, the first peak of the RDF curve corresponds to the first nearest-neighboring distance $r_{\rm{1NN}}$, providing a reference for selecting the value of $r_c$ (see the gray dashed lines in \autoref{fig:RDF-FCC-BCC-HCP}).
Typically, the atomic spacing at the "valley" of the RDF curve, which marks the transition region between two nearest-neighboring regions, is chosen as the core cutoff value $r_c$.
For the specific $r_c$ and $r_s$ values in \autoref{fig:RDF-FCC-BCC-HCP}, the average accumulated number of neighboring atoms within the core region, $\langle N_i^c \rangle$, and the total accumulated number of neighboring atoms, $\langle N_i^c+N_i^s \rangle$, are listed in \autoref{tab:neighboring_number}. 
These data provide strong evidence for subsequent efficiency analysis, further detailed in \autoref{sec:4.4}.

\begin{table}[htbp]
  \centering
  \caption{The cutoff radius and corresponding average accumulated numbers of neighboring atoms calculated by the RDF function.}\label{tab:neighboring_number}
  \begin{tabular}{c|cc|cc|cc}
  \toprule
   ~ &\multicolumn{2}{c|}{FCC} & \multicolumn{2}{c|}{BCC} & \multicolumn{2}{c}{HCP} \\
  \midrule   
   Radius (\AA) & $r_c = 2.8$ & $r_s = 4.95$ & $r_c = 2.7$ & $r_s = 5.3$ & $r_c = 3.5$ & $r_s = 7.5$ \\ 
     \hline  
   {Numbers of } & $\langle N_i^c \rangle$ & $\langle N_i^c +N_i^s \rangle$& $\langle N_i^c \rangle$ & $\langle N_i^c+N_i^s \rangle$& $\langle N_i^c \rangle$ & $\langle N_i^c+N_i^s \rangle$ \\ 
   neighboring atoms & 11.9964 & 42.5239  &  8.0345  & 58.0079 &  11.9646  &  71.0491\\ 
  \bottomrule
  \end{tabular}
\end{table}

Additionally, we further investigate the convergence behavior of the RBL-EAM method with respect to $r_c$ under a fixed batch size of $p=5$. Using $\alpha$-Fe metal as an example, 
we consider the  core cutoff radii $r_c = 2.5$, $3.0$, $3.8$, and $4.2$\AA. The corresponding RDF curves of the RBL-EAM method, shown in \autoref{fig:RDF-rc}, demonstrate rapid convergence to the reference EAM-DT curve as $r_c$ increases. 
These results confirm that with an appropriately chosen small batch size $p$ and core cutoff radius $r_c$, the RBL-EAM method can effectively capture the RDF with desired accuracy.
\begin{figure}[htbp]
    \centering
    \includegraphics[width = 0.95\textwidth]{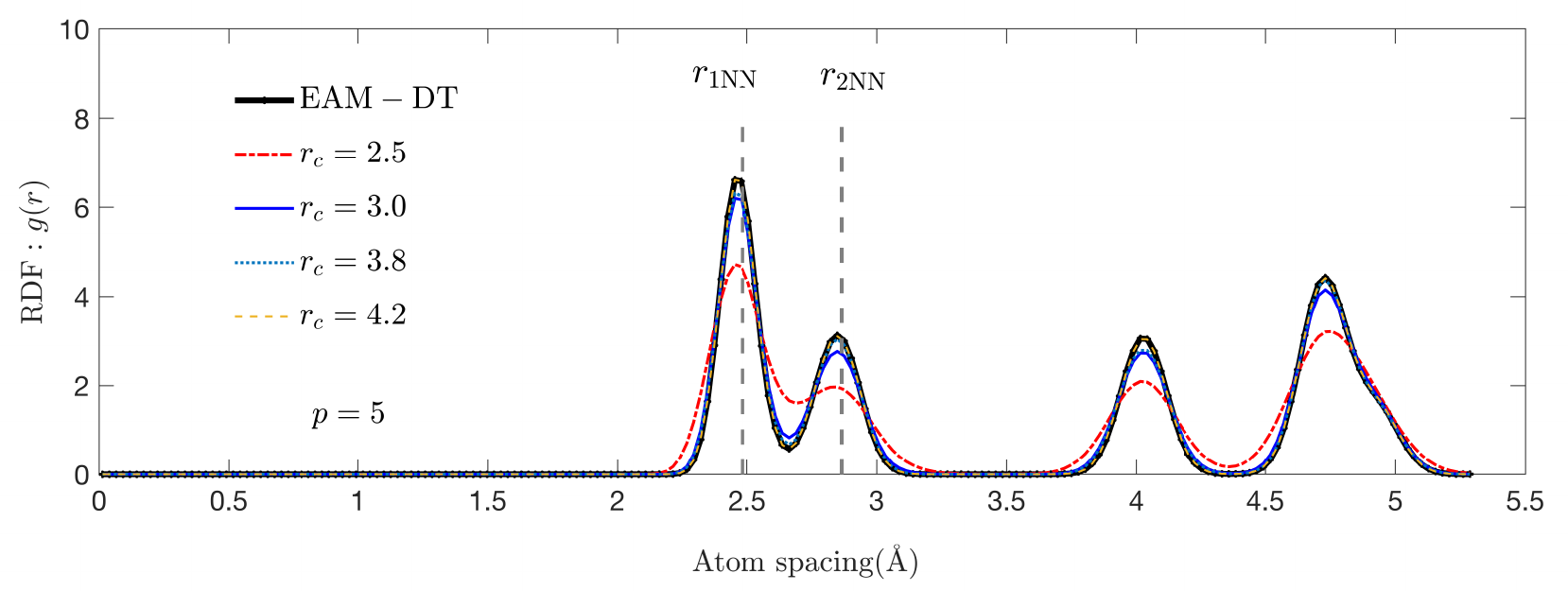}
    \caption{The RDF curves for $\alpha$-Fe with various core cutoff radii $r_c = 2.5$, $3.0$, $3.8$, $4.0$\AA\ and a fix batch size of $p=5$.}
    \label{fig:RDF-rc}
\end{figure}

\subsection{Elastic constants}\label{sec:4.3}
Elastic constants are fundamental parameters that characterize a material's stiffness in response to applied stresses. They play a crucial role in defining the elastic behavior of materials and are essential for solving linear elasticity equations. 
When experimental measurements of elastic constants are unavailable, MD simulations offer a reliable approach for their computation, serving as a bridge between atomic-scale interactions and continuum-level mechanical properties.
Elastic constants are determined by the interatomic forces acting on atoms when they are displaced from their equilibrium positions, thereby quantifying the material’s stiffness. The accuracy of these calculations is often used as a benchmark to assess the reliability of an interatomic potential. To evaluate the effectiveness of the proposed RBL-EAM method in practical applications, we employ it to compute the elastic constants of several metals.

For an atoms system in the NVT ensemble, the components $C_{\alpha \beta \mu \nu}$( $\alpha, \beta, \mu, \nu = 1, 2, 3$) of the elastic stiffness tensor $\bm{C}$  are calculated as follows \cite{ray1984statistical,shinoda2004rapid,clavier2017computation}:
\begin{equation}
  C_{\alpha \beta \mu \nu} = \langle C_{\alpha \beta \mu \nu}^B \rangle 
  - \frac{V}{k_B T}\Big[ \langle \varsigma_{\alpha \beta}^B \varsigma_{\mu \nu }^B \rangle - \langle \varsigma_{\alpha \beta}^B \rangle \langle\varsigma_{\mu \nu}^B\rangle\Big] 
  + \frac{N k_B T}{V} (\delta_{\alpha \mu}\delta_{\beta \nu} + \delta_{\alpha \nu}\delta_{\beta \mu}), 
\end{equation}
where $\bm{C}^B=(C_{\alpha \beta \mu \nu}^B)$ is the Born contribution to the elastic stiffness tensor, $\varsigma_{\alpha \beta}^B$ is the Born contribution to the stress tensor $\bm{\varsigma}$, $\delta_{\alpha\mu}$ is the Kronecker delta, and $V$ is the volume of the system.
The components of two tensors are given by:
\begin{equation*}
    C_{\alpha \beta \mu \nu}^B = \frac{1}{V} \frac{\partial^2 \mathcal{U}(\bm{q}_1, \bm{q}_2, \cdots \bm{q}_{N})}{\partial \epsilon_{\alpha \beta} \partial \epsilon_{\mu \nu}}, \quad
    \varsigma_{\alpha \beta}^B  = \frac{1}{V} \frac{\partial \mathcal{U}(\bm{q}_1, \bm{q}_2, \cdots \bm{q}_{N})}{\partial \epsilon_{\alpha \beta}}, 
\end{equation*}
where $\epsilon_{\alpha \beta}$ are the components of the strain tensor $\bm{\epsilon}$. 
Due to the symetry of $\bm{\varsigma}$ and $\bm{\epsilon}$, by introducing the Voigt notation (i.e., 11→1, 22→2, 33→3, 23→4, 13→5, 12→6), the generalized Hooke's law \cite{dell2009Hooke} can be written as $\varsigma_i = C_{ij} \epsilon_j$, $i,j = 1, 2, 3, 4, 5, 6$.
For the FCC and BCC crystals which exhibit cubic symmetry, there are only 3 independent elastic constants: $C_{11}$, $C_{12}$, $C_{44}$.
For the HCP crystals with hexagonal symmetry, the elastic tensor $\bm{C}$ has 6 independent elements, which are $C_{11}$, $C_{12}$, $C_{13}$, $C_{33}$, $C_{44}$, $C_{66}$.


%
To compute the elastic constants, we implement the explicit deformation method \cite{clavier2017computation}, which determines the set of elastic constants by loading six different deformations to the atomic system. 
The computational procedure is summarized in \autoref{Algorithm_4}, where the interaction forces and potentials are computed using the RBL-EAM method. 
We consider various batch sizes by $p = 0$, $5$, $10$, $15$, $20$ and $30$, respectively. 
Based on the results from RDF, we select two typical cutoff radii for each lattice type, i.e., 
$r_c = 2.8$, $4.0$\AA \ for the FCC lattice, $r_c = 2.7$, $3.1$\AA \ for the BCC lattice, and $r_c = 3.5$, $5.3$\AA \ for the HCP lattice, respectively. 
The shell cutoff radii are the same as those specified in the potential files. 
Each simulation is performed 5 times (except for $p=0$) to test the fluctuation of the elastic constants obtained by the RBL-EAM method. 
For comparison, we also run \autoref{Algorithm_4} for each simulation with EAM-DT method and take the resulting elastic constants as the reference values. 

\begin{algorithm}[H] 
\footnotesize{
  \caption{The computation process of elastic constants. }\label{Algorithm_4}
  \SetKwInOut{Initialize}{Initialize}
  \SetKwInOut{Compute}{Compute}
  \SetAlgoLined
  \KwIn{An EAM potential file for the given crystal.}
   \For{$ j = 1:6$}{
  \Initialize{The initial configuration and the specified temperature $T$.\\ Set the core cutoff radius $r_c$ and batch size $p$ for the RBL-EAM method.\\
  Relax the system in NVT ensemble for 2,000 steps;}
  \Compute{The ensemble average of stress tensor for initial configuration, denoted as $\bm{\varsigma}^0$.}
       Load a given elementary deformation $\epsilon_j$ and relax the system in NVT ensemble for 2,000 steps\;
       Run 1,000 steps to obtain the ensemble average stress $\bm{\varsigma}^j$ for the current configuration\;
       Calculate the components of elastic constants by the initial stress $\bm{\varsigma}^0$ and current stress $\bm{\varsigma}^j$.
    }
   \KwOut{The independent elastic constants of the given crystal.
   .
   } 
}
\end{algorithm}

\begin{figure}[htbp]
  \centering
  \subfloat[Cu-FCC]{\includegraphics[width = 0.8\textwidth]{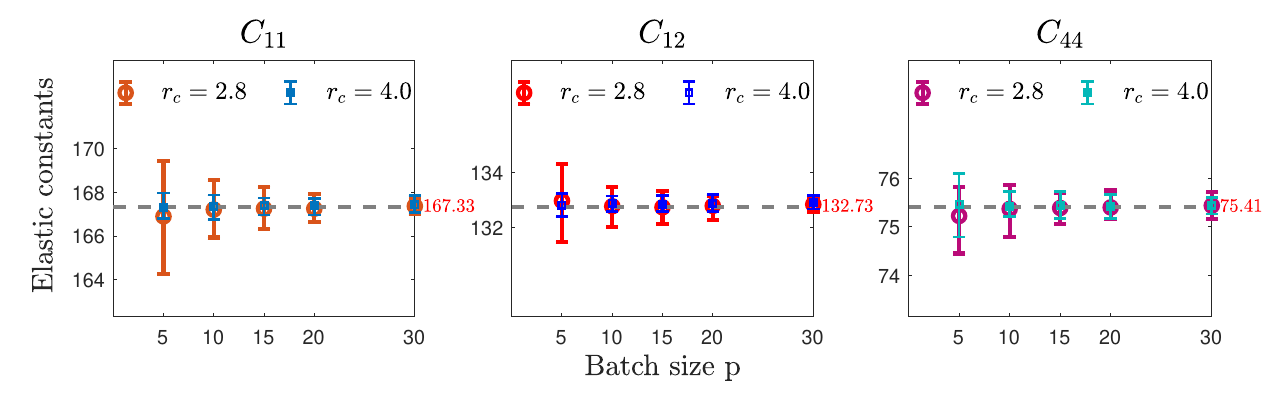}}\\
  \subfloat[$\alpha$-Fe-BCC]{\includegraphics[width = 0.82\textwidth]{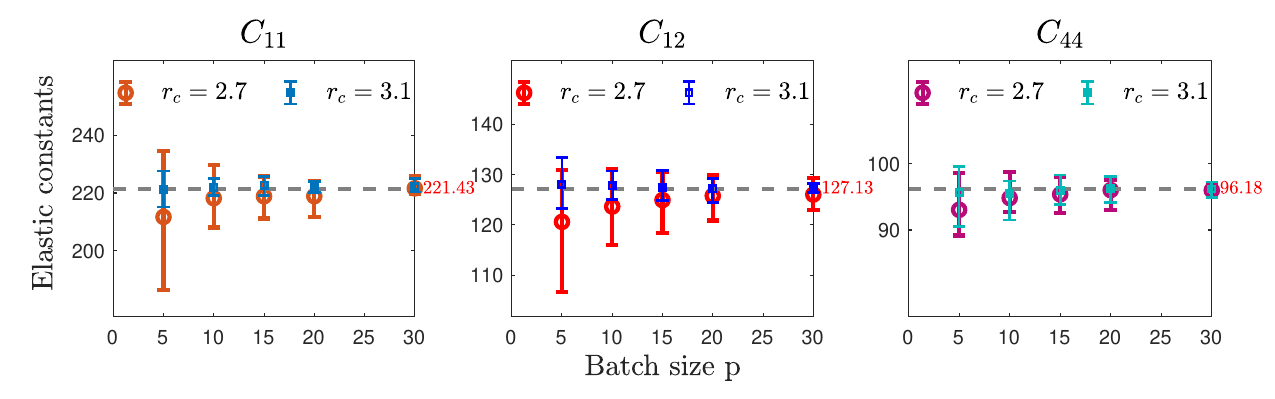}}\\
  \caption{Elastic constants $C_{11}$, $C_{12}$ and $C_{44}$ for the cubic crystals: (a) Cu and (b) $\alpha$-Fe. The results of RBL-EAM method with batch sizes $p = 5$, $10$, $15$, $20$, $30$ and two types of core cutoff radii are plotted. The reference values obtained by EAM-DT method is given at the right side. For easy comparison, we plot a dish line in these figures to identify the location of the reference value. 
  Similar phenomenon can be observed from the results of metal Mg. 
  }\label{fig:elastic_cubic}
\end{figure}
\begin{figure}[htbp]
  \centering
  \includegraphics[width = 0.82\textwidth]{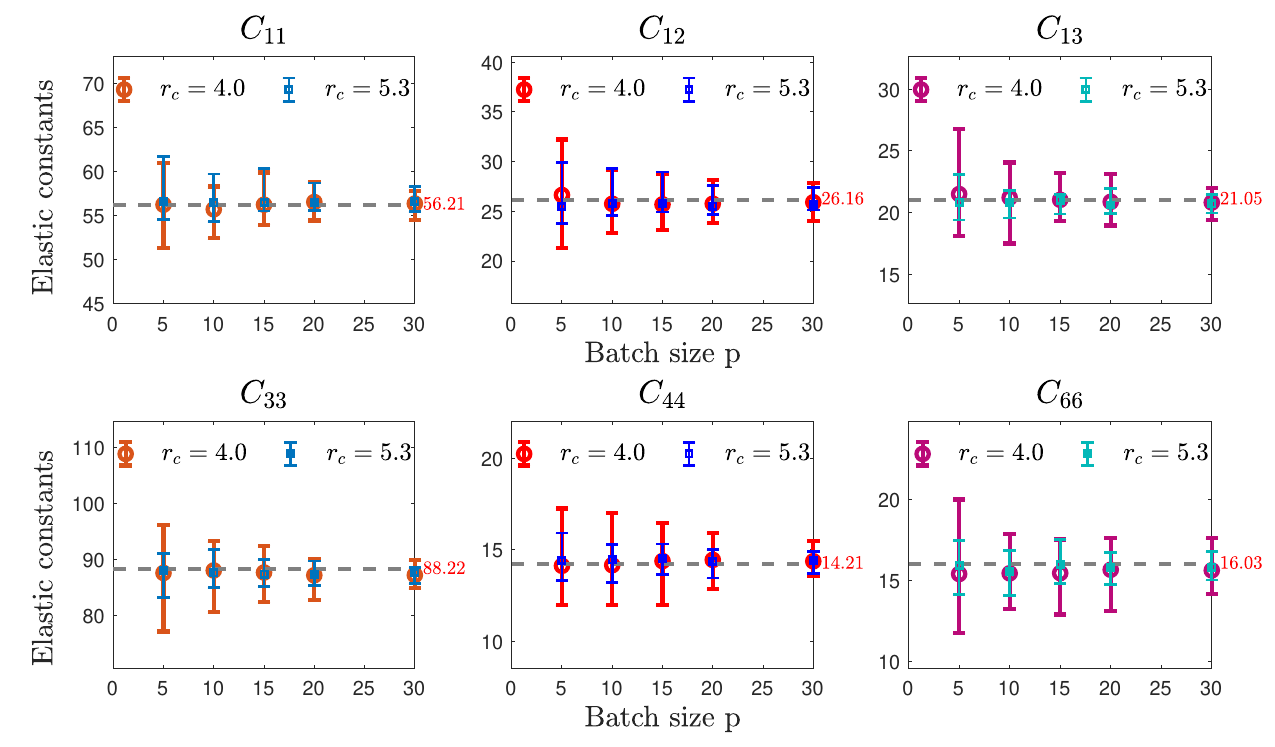}
  \caption{The elastic constants $C_{11}$, $C_{12}$, $C_{13}$, $C_{33}$, $C_{44}$ and $C_{66}$ for the hexagonal crystal Mg of HCP lattice, calculated by RBL-EAM method with batch sizes $p = 5, 10, 15, 20, 30$ with two core cutoff radii $r_c = 4.0$\AA\ and $r_c = 5.3$\AA and compared with the values (the right side in red) of EAM-DT method .}\label{fig:elastic_hcp}
\end{figure}

The results of elastic constants obtained by the RBL-EAM method are presented in \autoref{fig:elastic_cubic} and \autoref{fig:elastic_hcp} for Cu, $\alpha$-Fe, and Mg, respectively. 
The results shown in \autoref{fig:elastic_cubic} and Figure \autoref{fig:elastic_hcp} clearly demonstrate that the elastic constants computed using the RBL-EAM method with $p>0$ converge to the reference values as the batch size $p$ or core cutoff radius $r_c$ increases. 
Additionally, the stochastic fluctuation in the computed values decreases as the batch size or core cutoff radius increases.
For metal Cu, we suggest to use $p=5$ and $r_c=4.0$\AA \  or $p=10$ and $r_c=2.8$\AA. For metal $\alpha$-Fe,  if the batch size $p$ is set as 5, we should select a large core cutoff radius $r_c$, i.e., $r_c=3.1$\AA \  , to achieve an accurate elastic constants. 
We guess this is due to number ratio of neighboring atoms in the core and in the shell for $\alpha$-Fe is smaller than that of Cu, as reported in \autoref{tab:neighboring_number}. 
The elastic constants obtained by the RBL-EAM method with $p=0$ are listed in \autoref{tab:elastic_eam_0}.
From which, we observe that the elastic constants obtained by the RBL-EAM with $p=0$ exhibit a significant deviation from the reference values, which indicates that setting the cutoff radius $r_s=r_{\rm{1NN}}$ or $r_{\rm{2NN}}$ may lead to non-physical values for the elastic constants.

\begin{table}[htbp]\footnotesize
  \centering
  \caption{The elastic constants obtained by EAM-RBL method with $p=0$ or $p=5$ and EAM-DT method.}\label{tab:elastic_eam_0}
  \begin{tabular}{c|cc|ccccccc}
  \toprule
   ~ & ~ & ~& \multicolumn{6}{c}{Elastic constant (GPa)} \\
  Lattice  & {Methods} & Batch size & $c_{11}$ & $c_{12}$ & $c_{13}$ & $c_{33}$ & $c_{44}$ & $c_{66}$ \\
  \midrule   
   \multirow{3}*{FCC} & EAM-RBL & $p=0$ & 183.28 & 139.17 & - & - & 74.68 & - \\
                ~ &  ($r_c = 2.80$) & $p=5$ & 166.90 & 132.96 & - & - & 75.22 & - \\
                \cline{2-9}
     ~ & \makecell{EAM-DT\\($r_s = 4.94$)}&  & \textbf{167.33} & \textbf{132.73} & - & - & \textbf{75.41} & - \\
   \midrule                  
   \multirow{3}*{BCC} & {EAM-RBL} & $p=0$ & 120.39 & 158.02 & - & - & 20.18 & - \\
    ~ & ($r_c = 2.70$) & $p=5$ & 211.72 & 120.58 & - & - & 93.03 & - \\
    \cline{2-9}
     ~ & \makecell{EAM-DT\\($r_s = 5.30$)}&  & \textbf{221.43} & \textbf{127.13} & - & - & \textbf{96.18} & - \\
   \midrule  
   \multirow{3}*{HCP} & EAM-RBL & $p=0$ & 73.55 & 24.62 & 10.44 & 86.82 & 22.34 & 25.28 \\
     ~ &  ($r_c = 3.50$) & $p=5$ & 56.24 & 26.64 & 21.51 & 87.57 & 14.12 & 15.42 \\
     \cline{2-9}
     ~ & \makecell{EAM-DT\\($r_s = 7.50$)}& ~ & \textbf{56.76} & \textbf{25.93} & \textbf{21.16} & \textbf{87.61} & \textbf{13.92} & \textbf{15.83} \\
  \bottomrule
  \end{tabular}
\end{table}

\subsection{Efficiency of the RBL-EAM method}\label{sec:4.4}
Finally, we investigate the speedup of the proposed RBL-EAM method compared with the EAM-DT method. 
For solid materials with lattices under small deformations, atoms undergo random perturbations around their equilibrium positions. In this case, we can estimate the computational costs of the RBL-EAM and EAM-DT methods by considering the number of neighboring atoms in the core and shell regions.
The ratio of  computational costs  for these two methods should be close to $\frac{<N_i^c>+p}{N_i^c+N_i^s}$. Thus, the theoretical speedup of the RBL-EAM method is expected to be slightly less than $\frac{N_i^c+N_i^s}{<N_i^c>+p}$.
To verify the speedup, we simulate six systems with varying numbers of atoms $N=500$, $4,000$, $32,000$, $256,000$, $2,048,000$, $16,384,000$ using both RBL-EAM and EAM-DT methods, respectively. 
For each simulation, we run 20 steps in the NVT ensemble and measure the computational time for force and potential calculations to compute the speedup.
We consider two batch sizes $p=5$ and $10$. 
For metals Cu, $\alpha$-Fe, and Mg, the core cutoff radii are $r_c=2.8$, $2.7$ and $3.5$\AA, respectively. 

\begin{figure}[htbp]
  \centering
  \subfloat[Cu-FCC]{\includegraphics[width = 0.32\textwidth]{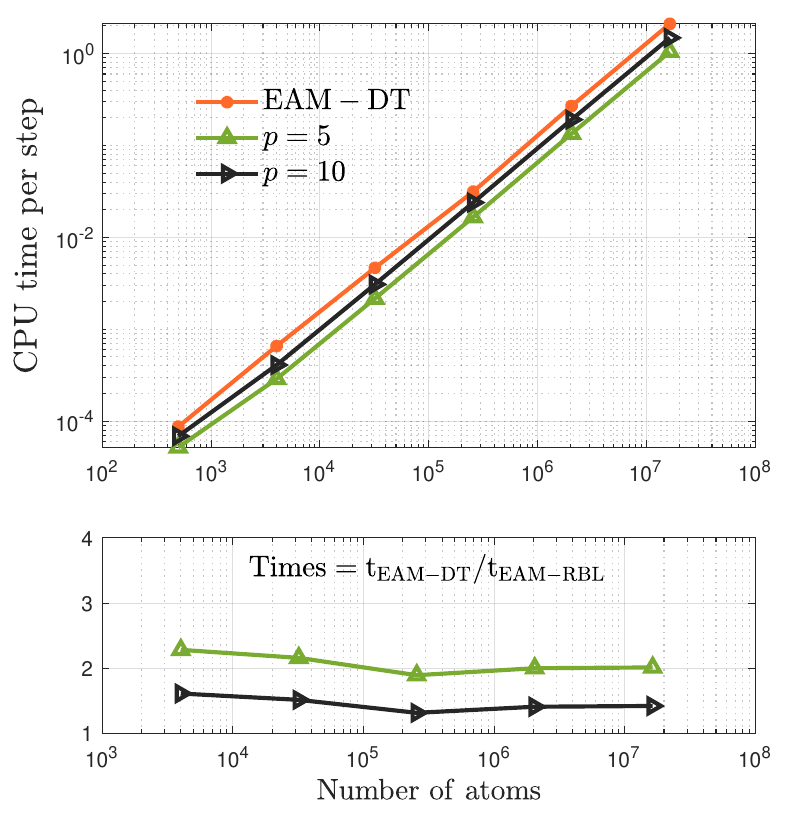}}
  \subfloat[$\alpha$-Fe-BCC]{\includegraphics[width = 0.32\textwidth]{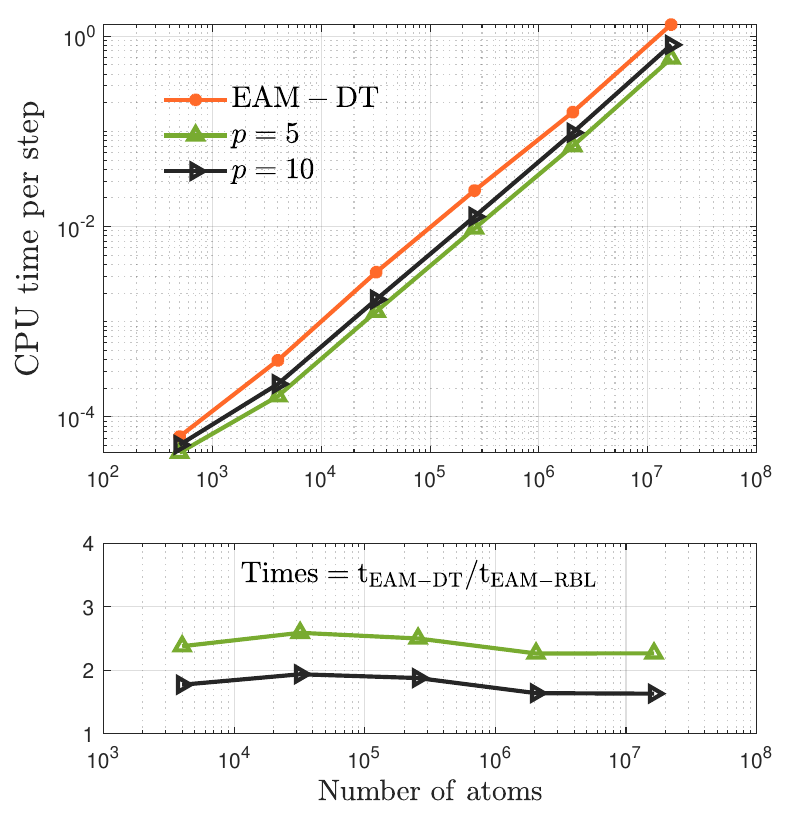}}
  \subfloat[Mg-HCP]{\includegraphics[width = 0.32\textwidth]{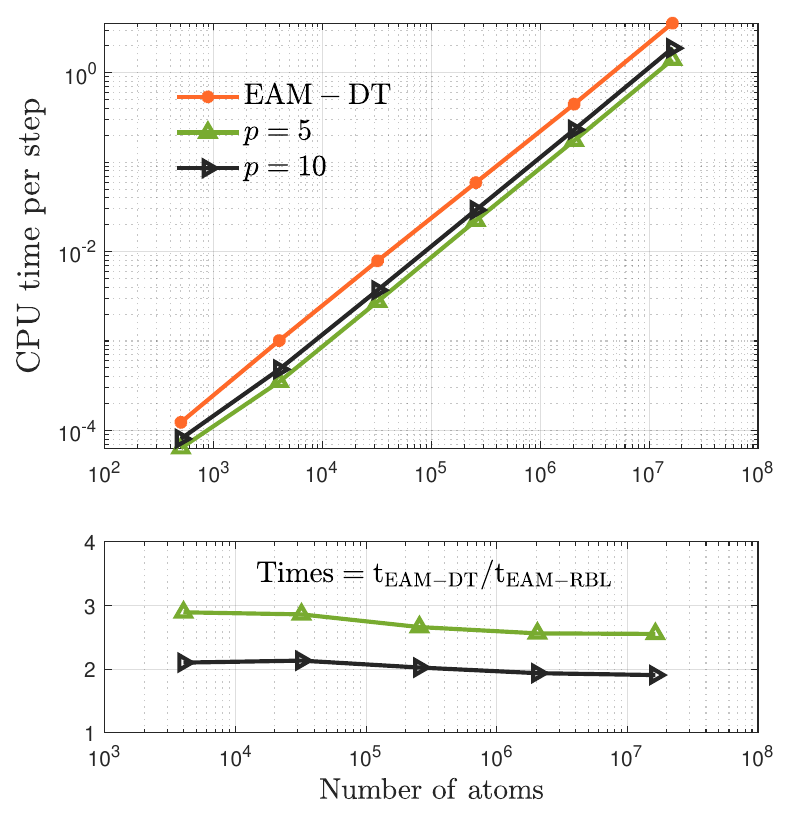}}
  \caption{The computational time (s) per time step for RBL-EAM and EAM-DT method. The speedup is computed by $t_{\rm{EAM-DT}}/t_{\rm{RBL-EAM}}$.}\label{fig:CPU_time}
\end{figure}

The computational times (in second per time step) and corresponding speedups for all simulations are reported in \autoref{fig:CPU_time}, which clearly shows that the RBL-EAM method achieves good speedup compared the EAM-DT method. 
Since both methods only involve near neighbor interaction, the computational time for these methods grows linearly with the increase of atoms number $N$.  
Furthermore, as shown in \autoref{fig:elastic_cubic}, the observed speedups of all simulations are close to the theoretical speedup. 
In summary, the RBL-EAM method is more efficient than the EAM-DT method.

%
\section{Conclusions} \label{sec:5}

In this paper, we have extended the RBL method to metallic systems with EAM potential. Leveraging the random mini-batch idea, both pairwise forces and embedding forces are approximately computed using two “core-shell” lists. The accuracy and efficiency of the proposed RBL-EAM method are validated through comprehensive numerical tests, including the calculation of lattice constants, radial distribution functions, and elastic constants. 
The RBL-EAM method achieves several times of the speedup compared to the traditional EAM-DT method, without compromising accuracy. 
This method shows great potential for large-scale simulations of metallic systems, with further improvements possible in batch size optimization and broader applications in material science, particularly in modeling complex atomic interactions and higher-order elastic constants.


\appendix
\section{The expectation of ${\mathcal{X}_i}$ in the case of $\mathfrak{B}_i'=\mathfrak{B}_i$.}\label{A-1}

For simplicity, let us consider a system with a full neighbor list, as used in \ref{forcefullsystem}. By taking $\mathfrak{B}_i'=\mathfrak{B}_i$, we have $H'_{ij}=H_{ij}$, which leads to $\mathbb{E}\left[ H'_{ij}H_{ik}\right]\neq \mathbb{E}\left[ H'_{ij}\right]\mathbb{E}\left[ H_{ik}\right]$.
Similar to \eqref{eq:26}, for $r_c<|\bm{r}_{ij}|\leq r_s$, we have 
\begin{equation}
\label{A-1-1}
\begin{aligned}\mathbb{E} \bigl[(\tilde{\bar{\rho}}_i &-{\bar{\rho}}_i) H_{ij}\bigr]= \frac{N_i^s}{p} \mathbb{E}\biggl[
      \sum_{r_c < |\bm{r}_{ik}| \leq r_s} \rho(|\bm{r}_{ik}|) H_{ij}H_{ik}
    \biggr]-\frac{p}{N_i^s}{\bar{\rho}}^s_i\\
&=\frac{N_i^s}{p} \mathbb{E}\biggl[\rho(|\bm{r}_{ij}|) H_{ij}+
      \sum_{\substack{ k(k\neq j)\\  r_c < |\bm{r}_{ik}| \leq r_s}} \rho(|\bm{r}_{ik}|) H_{ij}H_{ik}
    \biggr]-\frac{p}{N_i^s}{\bar{\rho}}^s_i\\
    &=\rho(|\bm{r}_{ij}|) +\frac{p-1}{N_i^s-1}
      \sum_{\substack{ k(k\neq j)\\  r_c < |\bm{r}_{ik}| \leq r_s}} \rho(|\bm{r}_{ik}|) 
    -\frac{p}{N_i^s}{\bar{\rho}}^s_i=\frac{N_i^s-p}{N_i^s-1}\rho(|\bm{r}_{ij}|) -\frac{N_i^s-p}{N_i^s(N_i^s-1)}{\bar{\rho}}^s_i,
\end{aligned}
\end{equation}
and 
\begin{equation}
\label{A-1-2}
\begin{aligned}\mathbb{E} \bigl[(\tilde{\bar{\rho}}_j &-{\bar{\rho}}_j) H_{ij}\bigr]= \frac{N_j^s}{p} \mathbb{E}\biggl[
      \sum_{r_c < |\bm{r}_{jk}| \leq r_s} \rho(|\bm{r}_{jk}|) H_{ij}H_{jk}
    \biggr]-\frac{p}{N_i^s}{\bar{\rho}}^s_j\\
&=\frac p{N_i^s} 
      \sum_{r_c < |\bm{r}_{jk}| \leq r_s} \rho(|\bm{r}_{jk}|)-\frac{p}{N_i^s}{\bar{\rho}}^s_j=0.
\end{aligned}
\end{equation}
By taking a summation of $j$ in \eqref{A-1-1}, we get
 \begin{equation}
 \label{A-1-3}
\sum_{\substack{j(j \neq i)\\  r_c < |\bm{r}_{ij}| \leq r_s}}\mathbb{E} \bigl[(\tilde{\bar{\rho}}_i -{\bar{\rho}}_i) H_{ij}\bigr]= \sum_{\substack{j(j \neq i)\\  r_c < |\bm{r}_{ij}| \leq r_s}}\frac{N_i^s-p}{N_i^s-1}\rho(|\bm{r}_{ij}|) -\frac{N_i^s-p}{N_i^s-1}{\bar{\rho}}^s_i=0. 
\end{equation}
Therefore, we have
\begin{equation}\label{A-2}    
\begin{aligned}\mathbb{E} \bigl[\tilde{\bm{\sigma}}^{e}_{ij,1} H_{ij}\bigr]&-\frac{p}{N_i^s}{\bm{\sigma}}^{e}_{ij,1}=-\mathbb{E} \bigl[\mathcal{F}'(\tilde{\bar{\rho}}_i) H_{ij}\bigr]{\rho}'(|\bm{r}_{ij}|)\frac{\partial |\bm{r}_{ij}|}{\partial \bm{q}_i}-\frac{p}{N_i^s}{\bm{\sigma}}^{e}_{ij,1}\\
&=-\mathbb{E} \bigl[\mathcal{F}'(\bar{\rho}_i)H_{ij} + \mathcal{F}''(\bar{\rho}_i) \cdot (\tilde{\bar{\rho}}_i -{\bar{\rho}}_i)H_{ij}\bigr]{\rho}'(|\bm{r}_{ij}|)\frac{\partial |\bm{r}_{ij}|}{\partial \bm{q}_i}-\frac{p}{N_i^s}{\bm{\sigma}}^{e}_{ij,1}\\
&=-\mathbb{E} \bigl[ (\tilde{\bar{\rho}}_i -{\bar{\rho}}_i)H_{ij}\bigr]\mathcal{F}''(\bar{\rho}_i) {\rho}'(|\bm{r}_{ij}|)\frac{\partial |\bm{r}_{ij}|}{\partial \bm{q}_i},
\end{aligned}
\end{equation}
and
\begin{equation}  
\label{A-3}
\begin{aligned}\mathbb{E} \bigl[\tilde{\bm{\sigma}}^{e}_{ij,2} H_{ij}\bigr]&-\frac{p}{N_i^s}{\bm{\sigma}}^{e}_{ij,2}=-\mathbb{E} \bigl[\mathcal{F}'(\tilde{\bar{\rho}}_j) H_{ij}\bigr]{\rho}'(|\bm{r}_{ji}|)\frac{\partial |\bm{r}_{ji}|}{\partial \bm{q}_i}-\frac{p}{N_i^s}{\bm{\sigma}}^{e}_{ij,2}\\
&=-\mathbb{E} \bigl[\mathcal{F}'(\bar{\rho}_j)H_{ij} + \mathcal{F}''(\bar{\rho}_j) \cdot (\tilde{\bar{\rho}}_j -{\bar{\rho}}_j)H_{ij}\bigr]{\rho}'(|\bm{r}_{ji}|)\frac{\partial |\bm{r}_{ji}|}{\partial \bm{q}_i}-\frac{p}{N_i^s}{\bm{\sigma}}^{e}_{ij,2}\\
&=-\mathbb{E} \bigl[ (\tilde{\bar{\rho}}_j -{\bar{\rho}}_j)H_{ij}\bigr]\mathcal{F}''(\bar{\rho}_j) {\rho}'(|\bm{r}_{ji}|)\frac{\partial |\bm{r}_{ji}|}{\partial \bm{q}_i}=0.
\end{aligned}
\end{equation}
By combining \eqref{eq:expect_ps}, \eqref{eq:expect_ec-1}, \eqref{A-2},  and \eqref{A-3}, we have
\begin{equation}
\label{A-4}
  \begin{split}
    \mathbb{E}[\bm{\mathcal{X}}_i] &=  \mathbb{E}\bigl[ \tilde{\bm{\sigma}}_i^{e} -{\bm{\sigma}}_i^{e} \bigr]\\&=\sum_{\substack{j(j \neq i)\\ |\bm{r}_{ij}| \leq r_c}} \mathbb{E} \left[ (\tilde{\bm{\sigma}}^{es}_{ij,1}+\tilde{\bm{\sigma}}^{es}_{ij,2})\right]+\frac{N_i^s}p\sum_{\substack{j(j \neq i)\\  r_c < |\bm{r}_{ij}| \leq r_s}} \mathbb{E} \left[ (\tilde{\bm{\sigma}}^{es}_{ij,1}+\tilde{\bm{\sigma}}^{es}_{ij,2}) H_{ij}\right]-{\bm{\sigma}}_i^{e}  \\&=-\frac{N_i^s}p\sum_{\substack{j(j \neq i)\\  r_c < |\bm{r}_{ij}| \leq r_s}}\mathbb{E} \bigl[ (\tilde{\bar{\rho}}_i -{\bar{\rho}}_i)H_{ij}\bigr]\mathcal{F}''(\bar{\rho}_i) {\rho}'(|\bm{r}_{ij}|)\frac{\partial |\bm{r}_{ij}|}{\partial \bm{q}_i}.
  \end{split}
\end{equation} 
Even with \eqref{A-1-3}, the last term in \eqref{A-4} is still nonzero because  ${\rho}'(|\bm{r}_{ij}|)\frac{\partial |\bm{r}_{ij}|}{\partial \bm{q}_i}$  differs for different $j$. Therefore, we have  $\mathbb{E}[\bm{\mathcal{X}}_i]=0$ does not hold, even for a linear function ${\mathcal{F}'(\rho)}$.



\bibliographystyle{elsarticle-num}
\bibliography{ref.bib}

\end{document}